\documentclass[superscriptaddress,aps,twocolumn,longbibliography,nofootinbib,pra]{revtex4-2}
\usepackage{graphicx, color, graphpap}      
\usepackage{amsmath}
\usepackage{bibunits}
\defaultbibliography{refs.bib}
\defaultbibliographystyle{unsrt}
\usepackage{enumerate}
\usepackage{amssymb}
\usepackage{amsthm}
\usepackage{pstricks}
\usepackage{float}
\usepackage[colorlinks=true, citecolor=blue, linkcolor=red]{hyperref}
\usepackage[T1]{fontenc}
\usepackage{bbm}
\usepackage{dsfont}
\usepackage[linesnumbered,ruled,vlined]{algorithm2e}
\SetKwInput{kwInit}{Init}
\usepackage{mathtools}

\usepackage{tikz}
\usepackage{graphicx}
\usetikzlibrary{positioning}
\usepackage{graphicx}
\usepackage{xcolor}
\usepackage[caption=false]{subfig}
\usepackage{siunitx}
\usepackage{pifont} 
\usepackage{todonotes}
\usepackage{physics}
\usepackage{thmtools}
\usepackage{thm-restate}
\usepackage{quantikz}
\usetikzlibrary{calc}

\renewcommand{\ket}[1]{| #1 \rangle}

\newcommand{\ovl}[1]{\overline{#1}}
\newcommand{\Rcal}{\mathcal{R}}
\newcommand{\Gcal}{\mathcal{G}}
\newcommand{\Vr}[1]{\text{Var} \left( #1 \right)}

\renewcommand{\eqref}[1]{(\ref{#1})}

\newtheoremstyle{example}{\topsep}{\topsep}%
{}
{}
{\bfseries}
{:}
{   }
{\thmname{#1}\thmnumber{ #2}}
\theoremstyle{example}
\newtheorem{theorem}{Theorem}
\newtheorem{lemma}{Lemma}
\newtheorem{corollary}{Corollary}

\newtheorem{proposition}{Proposition}

\theoremstyle{definition}
\newcommand{\G}{\mathcal{G}}

\newtheorem{definition}{Definition}

\newtheorem*{theorem*}{Theorem}

\def\orcid#1{\kern -0.4em\href{https://orcid.org/#1}{\includegraphics[keepaspectratio,width=0.7em]{orcid_logo.pdf}}}

\usepackage{lipsum}

\long\def\ca#1\cb{} 
\begin{document}
\title{Overlapped groupings for quantum energy estimation:\texorpdfstring{\\}{ }Maximal variance reduction and deterministic algorithms for reducing variance}

\author{Jeremiah Rowland}
\affiliation{Department of Physics and Astronomy, Michigan State University, East Lansing, MI 48823, USA}
\affiliation{Center for Quantum Computing, Science, and Engineering, Michigan State University, East Lansing, MI 48823, USA}

\author{Rahul Sarkar}
\affiliation{Department of Mathematics, University of California, Berkeley}

\author{Nicolas PD Sawaya}
\affiliation{Azulene Labs, Berkeley, CA 94720}

\author{Norm M. Tubman}
\affiliation{Applied Physics Group, NASA Ames Research Center}

\author{Ryan LaRose} \thanks{Corresponding author:  \href{rmlarose@msu.edu}{rmlarose@msu.edu}}
\affiliation{Department of Computational Mathematics, Science, and Engineering, Michigan State University, East Lansing, MI 48823, USA}
\affiliation{Department of Electrical and Computer Engineering, Michigan State University, East Lansing, MI 48823, USA}
\affiliation{Department of Physics and Astronomy, Michigan State University, East Lansing, MI 48823, USA}
\affiliation{Center for Quantum Computing, Science, and Engineering, Michigan State University, East Lansing, MI 48823, USA}

\begin{abstract}
        Grouping-based measurement strategies are widely used to reduce measurement complexity in near-term quantum algorithms. While these schemes have typically produced disjoint groups, recently this has been relaxed in what is known as overlapped grouping or coefficient splitting where operators may appear in more than one compatible group. In recent work, it has been numerically shown that this strategy can reduce the variance of energy estimates on small benchmark problems, motivating both the application and further analysis of the method. Here we prove that overlapped grouping for energy estimation can lead to a maximal variance reduction that is linear in the number of Hamiltonian terms. We introduce a new algorithm which we call repacking to transform existing groups into overlapped groups, and we show this repacking procedure iteratively reduces variance under mild assumptions. We also perform numerical simulations with  Hamiltonians up to $44$ qubits and $575 \cdot 10^{3}$ terms, assessing overlapped grouping at scale on problems of practical importance. Our numerics show that the variance reduction relative to state-of-the-art (disjoint) grouping increases linearly with the problem size, suggesting that overlapped grouping methods can be a powerful strategy for quantum energy estimation at the scale of Megaquop computers and beyond.
\end{abstract}

\maketitle


\section{Introduction}
Measurement complexity is a dominant contributor to the cost of quantum algorithms, determining both the number of distinct circuits and the number of shots required per circuit to reach a target accuracy.  Particularly pronounced in near-term quantum algorithms, measurement complexity has been well-studied in variational quantum algorithms~\cite{peruzzoVariationalEigenvalueSolver2014, mccleanTheoryVariationalHybrid2016, weckerProgressPracticalQuantum2015, verteletskyiMeasurementOptimizationVariational2020} where a prominent proposal to mitigate measurement complexity is grouping-based schemes~\cite{verteletskyiMeasurementOptimizationVariational2020,zhaoMeasurementReductionVariational2020,crawfordEfficientQuantumMeasurement2021}. These schemes reduce measurement complexity by grouping Pauli operators into disjoint sets of mutually commuting terms which can then be measured simultaneously in a shared basis.  Grouping has been shown to be an effective near-term strategy for producing energy estimates for moderately-sized Hamiltonians such as those arising in chemistry and condensed matter physics~\cite{crawfordEfficientQuantumMeasurement2021}.  Other proposed strategies for reducing measurement complexity include randomized~\cite{huangPredictingManyProperties2020, hadfieldMeasurementsQuantumHamiltonians2022} or derandomized~\cite{huangEfficientEstimationPauli2021} measurements, efficiently computing $k$-RDMs~\cite{BonetMonroig_Babbush_OBrien_2020},  informationally complete measurements~\cite{Garcia_Rossi_Sokolov_Tacchino_Barkoutsos_Mazzola_Tavernelli_Maniscalco_2021}, machine learning based methods~\cite{Torlai_Mazzola_Carleo_Mezzacapo_2020,Quek_Fort_Ng_2021,Smith_Gray_Kim_2021}, as well as methods tailored to particular physical systems such as fermionic Hamiltonians~\cite{hugginsEfficientNoiseResilient2021,choiFluidFermionicFragments2023}.

While grouping-based schemes have typically considered disjoint groups,  this is not a fundamental constraint. As an example, for the grouping $\{\{ZI,ZZ\},\{XX\}\}$, the operator $ZZ$ is measured only in the first group, but $ZZ$ also commutes with $XX$. Thus, $ZZ$ is compatible with both groups, allowing one to extract additional information with the same number of shots. This idea has recently been explored under the names of overlapped grouping and coefficient splitting~\cite{wuOverlappedGroupingMeasurement2023,yenDeterministicImprovementsQuantum2023,greschGuaranteedEfficientEnergy2025}, and to our knowledge was first suggested in~\cite{crawfordEfficientQuantumMeasurement2021} under the name of coefficient splitting. In particular, Ref.~\cite{wuOverlappedGroupingMeasurement2023} used this idea in numerical benchmarks for molecular Hamiltonians with up to $16$ qubits and showed improvements relative to other measurement strategies. Ref.~\cite{yenDeterministicImprovementsQuantum2023} introduced two overlapped grouping strategies and showed reduced variance for molecular Hamiltonians including H$_2$, LiH, BeH$_2$, and NH$_3$ relative to previous methods. Last, Ref.~\cite{greschGuaranteedEfficientEnergy2025} combined the idea of overlapped groupings with shadow tomography to again show variance reduction of small benchmark molecules (up to $16$ qubits). A similar idea was used in~\cite{choiGhost2022} in which ``ghost'' operators appearing in multiple groups are introduced to decrease variance. Based on the initial numerical success of this approach, it is natural and important to ask about the maximal variance reduction possible for overlapped grouping methods, when overlapped groupings can reduce measurement complexity, and how effective these methods will be for practical problems at the scale of Megaquop computers~\cite{Preskill_2025} and beyond.

In this work we provide answers to these questions. First, we consider the maximal variance reduction of overlapped grouping relative to (disjoint) grouping. This reduction of course depends on the original grouping. For this, we consider the widely-used sorted insertion~\cite{crawfordEfficientQuantumMeasurement2021} heuristic for grouping, and construct an explicit Hamiltonian where the overlapped grouping reduction in variance is proportional to the number of groups in the original disjoint grouping. This reduction matches the maximum possible asymptotic variance reduction arising from overlapping groupings under a fixed allocation strategy relative to the original strategy. Our only assumption in this proof is that covariance terms vanish, an assumption that has been employed when constructing heuristic estimators for (overlapped) grouping variance~\cite{wuOverlappedGroupingMeasurement2023,yenDeterministicImprovementsQuantum2023}. In our discussion, we present examples where this estimator is accurate, as well as counterexamples. In our proof, we introduce a new algorithm to produce overlapped groupings from (disjoint) groupings which we call \textit{repacking}. 
Under mild assumptions on the Hamiltonian, we show that repacking always reduces the variance of grouping-based methods. Thus, repacking provides a deterministic way to reduce variance relative to an initial grouping when estimating energy. 
Finally, we augment these analytical results with numerical results on overlapped groupings for practical problems. Notably, we perform numerics with electronic structure Hamiltonians that have up to $44$ qubits and $575 \cdot 10^3$ terms, significantly beyond what has been done numerically in previously-mentioned works. Relative to sorted insertion, we corroborate our analytical result that overlapped grouping reduces variance in all of our numerical experiments, up to a factor of ${\sim}2.35$x for our largest molecular benchmarks, and empirically find that the variance reduction increases in the size of the problem. In addition to our analytical results, this suggests that overlapped grouping methods can provide significant savings for quantum energy estimation with practical problems at scale.

The rest of the paper is organized as follows. We first present a summary our main results in Sec.~\ref{sec:main-results}. In Sec.~\ref{sec:definitions}, we review preliminary definitions and prior work. In Sec.~\ref{sec:methods}, we introduce our new methods including empirical estimators for overlapped groupings (Sec.~\ref{sec:empirical-estimators-for-overlapped-groupings}), optimal shot allocation for overlapped groupings (Sec.~\ref{sec:optimal-shot-allocation}), and repacking algorithms to form overlapped groupings from (disjoint) groupings (Sec.~\ref{sec:methods-repacking}). As we describe, repacking can be done without changing measurement bases, which we call a \textit{post hoc repacking}, and thus can be implemented with samples obtained from previous experiments. Or, repacking can be done by changing measurement bases prior to experimental sampling, which we call \textit{ad hoc repacking}. In Sec.~\ref{sec:results} we formally state and prove our analytical results (Sec.~\ref{subsec:analyticalResults}) and present our numerical results (Sec.~\ref{sec:numerical-results}). 

\section{Summary of Main Results} \label{sec:main-results}

\begin{figure}
    \centering
    \includegraphics[width=\columnwidth]{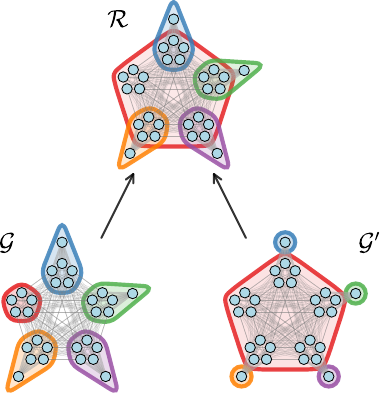}
    \caption{
        Illustration of the family of Hamiltonians which admit the maximal variance reduction of overlapped groupings relative to disjoint groupings, as stated in Theorem~\ref{thm:maximal-variance-reduction}. As defined in Sec.~\ref{sec:definitions}, nodes represent Hamiltonian terms, and terms that commute share an edge. Two (disjoint) groupings, $\G$ and $\G'$, as well as the overlapped grouping $\Rcal$ found by repacking, are shown for the same Hamiltonian (graph). This Hamiltonian has $n$ complete graphs $K_n$ which are interconnected, as well as $n - 1$ ``exterior nodes'' interconnected with the nearest complete graph. 
    }
    \label{fig:partial_order}
\end{figure}

Here we summarize the primary contributions of our work. First, we address the key question: What is the best case improvement for overlapped grouping methods relative to (disjoint) grouping methods? It is clear that this depends on the particular grouping method chosen, as one could artificially construct a bad grouping (for example, the trivial grouping in which every Hamiltonian term is placed into its own group), for which the improvement from overlapped grouping would be artificially large. To avoid this, we consider the performant sorted insertion~\cite{crawfordEfficientQuantumMeasurement2021} heuristic algorithm to form a disjoint grouping. Sorted insertion works by sorting Hamiltonian terms by coefficient weight and then placing terms into the first compatible group, forming a new group if there are no compatible groups. This method has proven performant on electronic structure Hamiltonians and a variety of other benchmarks~\cite{crawfordEfficientQuantumMeasurement2021,dalfaveroMeasurementReductionExpectation2025}. We show that there exists a family of Hamiltonians for which overlapped groupings outperform disjoint groupings found by sorted insertion by a linear factor in the number of groups.

\begin{theorem} \label{thm:maximal-variance-reduction}
There exist Hamiltonians and states for which the variance reduction of overlapped grouping relative to the best disjoint grouping found by sorted insertion is $\Theta(m)$, where $m$ is the number of disjoint groups. Assuming a model of state-independent variance and zero covariance, this is the asymptotic maximal variance reduction that can be achieved by overlapped grouping. 
\end{theorem}

The commutativity structure of the family of Hamiltonians which admit this variance reduction is shown in Fig.~\ref{fig:partial_order}. These Hamiltonians consist of $n$ complete graphs $K_n$, with each node sharing an edge to every other node in a complete graph. (As defined in Sec.~\ref{sec:definitions}, we use the usual representation where each node is a Hamiltonian term, and commuting terms have an edge between them.) Additionally, all but one complete graphs have an ``exterior node'' with edges to all nodes in the nearest complete graph. The bottom panel shows two disjoint groupings $\G$ and $\G'$ found by sorted insertion, while the top panel shows the optimal overlapped grouping $\Rcal$. Intuitively, overlapped grouping allows for all exterior nodes to be grouped with their nearest complete graph (as in $\G$) \textit{and} for all complete graphs to be grouped (as in $\G'$), which is not possible with a disjoint grouping. The formal statement and proof of Theorem~\ref{thm:maximal-variance-reduction} is given in Sec.~\ref{subsec:analyticalResults}. We emphasize from this formal statement that Hamiltonian coefficients are all chosen to be $O(1)$ to purposefully find a challenging case for sorted insertion so that we may characterize the maximal variance reduction we can hope to obtain by overlapped groupings. In addition to this model, we numerically explore molecular Hamiltonians and certain models of random Hamiltonians to characterize the average-case improvement from overlapped groupings, as discussed later in this section.

In Fig.~\ref{fig:partial_order} and the above discussion, we use $\Rcal$ for the overlapped grouping to emphasize it is obtained through a particular algorithm we introduce which we call repacking. While there can be many strategies to find an overlapped grouping, repacking starts from a given disjoint grouping, then seeks to add terms to additional groups in a specified manner. We introduce two such algorithms, {post-hoc repacking} and {ad-hoc repacking}. These repacking algorithms have the key property that they only add operators to existing measurement groups and do not form any new groups.  
Under mild conditions on the Hamiltonian, and for a suitable variance estimator used in prior work that we employ here, we show that repacking algorithms strictly decrease variance. 

\begin{theorem} \label{thm:repacking-reduces-variance}
Let $\mathcal{G}$ be a grouping of a Hamiltonian $H$ and let $\mathcal{R}$ be a repacking of $\mathcal{G}$.  Assume that for all distinct Pauli operators $P_i, P_k$ in $\mathrm{supp}(H)$, $\sigma_{P_i P_k} = 0.$
Fix a shot allocation $\{M_j\}_{j=1}^m$ over the groups such that group $G^{[j]}$ and its corresponding repacked group $G'^{[j]}$ are measured $M_j$ times.  Using a shot-weighted averaging for the energy estimators $\ovl{E}_{\Gcal}$ and $\ovl{E}_{\Rcal}$, we have that 
\begin{equation}
    \Vr{\ovl{E}_{\Rcal}} < \Vr{\ovl{E}_{\Gcal}}
\end{equation}
if at least one refinement step adds a Pauli operator $P_s$ with $\sigma^2_{P_s} > 0$.
\end{theorem}
Intuitively, this result says that overlapped groupings found through our repacking algorithms will always decrease variance, and so reduce the total number of measurements required to estimate energy to a given accuracy. This behavior may not be true with other methods to find overlapped groupings, especially when compared to optimal disjoint groupings found through well-performing heuristics like sorted insertion. Our repacking algorithms thus present a promising strategy to reduce measurement complexity. The proof of Theorem~\ref{thm:repacking-reduces-variance} is given in Sec.~\ref{subsec:analyticalResults}.

Finally, we perform large-scale numerics with Hamiltonians that have up to $44$ qubits and $575 \cdot 10^3$ (Pauli) terms, significantly beyond what has been done numerically in previous works. These results, shown in Sec.~\ref{sec:numerical-results}, characterize the average-case performance of overlapped grouping methods. We find that repacking reduces variance up to ${\sim}2.35$x for molecular Hamiltonians, with a larger reduction generally occurring for larger problem sizes. We also study random Hamiltonian models and find that the repacking variance reduction grows with the problem size, approximately linearly. These results suggest that overlapped grouping methods, especially those found by repacking, can provide a substantial reduction to measurement complexity on problem sizes suitable for Megaquop computers~\cite{Preskill_2025} --- concretely, for energy estimation problems defined on $10^2$ --- $10^3$ qubits.

\section{Preliminaries:\texorpdfstring{\\}{ }Notation and Definitions} \label{sec:definitions}

We express an $n$-qubit Hamiltonian in the Pauli basis 
\begin{equation}
\label{eq:hamiltonian-def}
H = \sum_{i = 1}^{N} c_i P_i,
\qquad
P_i = \bigotimes_{j=1}^n p_j,\; p_j \in \{I, X, Y, Z\},
\end{equation}
and, given a state $|\psi\rangle$, write its exact energy as
\begin{equation}
    E \equiv E_{|\psi\rangle} = \langle H \rangle \equiv \langle \psi | H | \psi \rangle = \sum_i c_i \langle P_i \rangle.
\end{equation}
Typically we will simply write $E$ and $\langle H \rangle$, leaving the state $|\psi\rangle$ implicit.  
The Pauli support of $H$, $\text{supp}(H)$, is the set of Pauli operators with non-zero coefficients when $H$ is expressed in the Pauli basis, i.e. the Pauli operators  in~\eqref{eq:hamiltonian-def} for which $c_i \neq 0$.

\subsection{Empirical estimators}

Throughout this work, we distinguish between true expectation values $\langle P_i \rangle$ and empirical estimators obtained from a finite number of measurement shots.  The true expectation values are fixed quantities determined by the target state $|\psi\rangle$ via $\langle P_i \rangle = \langle \psi | P_i | \psi \rangle$, while empirical estimators are random variables arising from measurement outcomes.  We denote empirical estimators using an overbar. Given $M$ samples (bitstrings) from a state $|\psi\rangle$ collected from a measurement basis in which a Pauli operator $P$ is diagonal, its empirical estimator is given by
\begin{equation}
\overline{\langle P \rangle} = \frac{1}{M} \sum_{k = 1}^{M} x_k,
\end{equation}
where $x_k \in \{ \pm 1 \}$ is the $k$th measurement outcome. With $M$ samples per Pauli, the empirical estimator for energy is uniquely defined as
\begin{equation}
    \overline{E} = \overline{\langle H \rangle} = \sum_{i = 1}^{N} c_i \overline{\langle {P_i} \rangle}.
\end{equation}

Since $P^2=I$ for Pauli operators, their variance is
\begin{equation}
    \text{Var}(P_i) = \sigma_{P_i}^2 = 1 - \langle P_i \rangle ^2 .
\end{equation}
The corresponding variance of the empirical estimator $\overline{\langle P_i \rangle}$ with $M$ samples is
\begin{equation}
    \text{Var}\left(\overline{\langle P_i \rangle}\right)
    =
    \frac{\sigma_{P_i}^2}{M}.
\end{equation}
If two Pauli operators $P_i$ and $P_j$ are measured simultaneously, their covariance is
\begin{equation}
\text{Cov}(P_i, P_j)
=
\sigma_{P_i P_j}
=
\langle P_i P_j \rangle
-
\langle P_i \rangle \langle P_j \rangle.
\end{equation}
If $P_i$ and $P_j$ are measured simultaneously $M$ times, the covariance of their empirical estimators is
\begin{equation}
\text{Cov}\left(
\overline{\langle P_i \rangle},
\overline{\langle P_j \rangle}
\right)
=
\frac{\sigma_{P_i P_j}}{M}.
\end{equation}
The variance of H is then given by

\begin{equation} \label{eqn:variance-exact}
\text{Var}(H) =
\sum_i c_i^2\, \text{Var}(P_i)
+\, 2 \sum_{i<k} c_i c_k\,
\text{Cov}\!\left(
    P_i,P_k
\right).
\end{equation}

We remark that although all examples presented in this work use the Pauli basis, none of our methods rely on structure unique to the Pauli operators.  Indeed, our algorithms for forming overlapped groupings (Sec.~\ref{sec:methods-repacking}) require only that the composite observable be decomposable into groups of jointly measurable operators.  For concreteness, and because this setting encompasses many near-term quantum simulation tasks, we specialize the notation and adopt the corresponding Pauli-specific estimators.

\subsection{Measurement complexity}

The goal of quantum energy estimation is to estimate the energy to a specified accuracy $\epsilon > 0$. Summarizing from~\cite{crawfordEfficientQuantumMeasurement2021}, the measurement complexity (total number of measurements) required for accuracy $\epsilon$ with $m$ groups is
\begin{equation}
    \label{eqn:measurements-for-accuracy-epsilon-with-variance}
    M = \left( \frac{1}{\epsilon} \sum_{i = 1}^{m} \sqrt{\text{Var}(O_i)} \right)^2
\end{equation}
where $O_i$ is the operator corresponding to the $i$th group
\begin{equation}
    O_i := \sum_{j} c^{[i]}_j P^{[i]}_j .
\end{equation}
This assumes that measurements are optimally distributed --- i.e., given a fixed measurement budget, the number of measurements per group is allocated so as to optimally reduce variance. Throughout the paper, we use ``reduce variance'' and ``reduce measurement complexity'' synonymously in light of~\eqref{eqn:measurements-for-accuracy-epsilon-with-variance}.

\subsection{Groupings and overlapped groupings} \label{sec:groupings-and-overlapped-groupings}
Grouping-based measurement procedures estimate a composite observable $H$ by partitioning operators into mutually commuting sets or groups.

\begin{definition}[Grouping] \label{def:grouping}
    A \textbf{grouping} $\G$ of a Hamiltonian $H = \sum_{i = 1}^{N} c_i P_i$ is a collection of disjoint sets of commuting Pauli operators in the Pauli support of $H$, whose union covers all Pauli operators in $H$. That is, $\G$ can be written
    \begin{equation} \label{eqn:grouping-definition}
        \G = \left\{ G^{[1]}, ..., G^{[m]} \right\}, \quad G^{[i]} = \left\{
            P^{[i]}_1, P^{[i]}_2, ..., P^{[i]}_{N_i} \right\}
    \end{equation}
    where $ G^{[i]} \cap G^{[j]} = \emptyset$ for all $i \neq j$ (disjoint), $\bigcup_{i = 1}^{m} G^{[i]} = \{ P_i \}_{i = 1}^N$ (covering), and $[P^{[i]}_j, P^{[i]}_k] = 0$ for all $j, k = 1, ..., N_i$ (commuting). We refer to each $G^{[i]}$ as a  \textbf{group}, and say the grouping has $m$ groups.
\end{definition}

Once a grouping $\G$ is obtained, the experiment proceeds by measuring the terms in each group by rotating into their common eigenbasis. The variance associated with a grouping arises from the estimator used to combine measurement outcomes across groups, and depends on both the shot allocation and the covariance structure of simultaneously measured observables. We make this dependence explicit in Sec.~\ref{sec:methods}, where we define the estimator and derive its variance.

A standard technique for finding a grouping formulates the problem in terms of the commutativity graph $G = (V, E)$ where each vertex corresponds to a Pauli operator $P_i$ and edges connect pairs of operators that commute.  A group is any clique in this graph, representing a set of operators in $H$ which are simultaneously diagonalizable. The problem of finding a grouping is equivalent to that of finding a disjoint clique covering of $G$ or a coloring of the complement graph $G^c$~\cite{verteletskyiMeasurementOptimizationVariational2020}. Finding a minimum-variance or minimum-size disjoint clique covering is NP-Hard in general~\cite{gareyComputersIntractabilityGuide1990}.  Practical methods therefore rely on greedy heuristics.  In the standard greedy algorithm, operators are processed in a fixed order and inserted into the first compatible existing group. If no existing group is compatible, a new group is created. The commonly used sorted insertion heuristic orders operators by decreasing coefficient magnitude, which typically reduces the worst-case variance relative to a random permutation~\cite{crawfordEfficientQuantumMeasurement2021}.  Although this ordering performs well in practice, it is not globally optimal and can exhibit poor behavior on pathological instances.  The classical runtime of this algorithm is $O(|V|^2)$, corresponding to the worst-case number of pairwise commutativity checks.

An overlapped grouping simply relaxes the condition that groups are disjoint.

\begin{definition}[Overlapped grouping] \label{def:overlapped-grouping}
    An \textbf{overlapped grouping} $\G$ of a Hamiltonian $H = \sum_{i = 1}^{N} c_i P_i$ is a list of (not necessarily disjoint) sets of commuting Pauli operators in the Pauli support of $H$, whose union covers all Pauli operators in $H$. That is, $\G$ can be written
    \begin{equation} \label{eqn:overlapped-grouping-definition}
        \G = \left\{ G^{[1]}, ..., G^{[m]} \right\}, \quad G^{[i]} = \left\{
            P^{[i]}_1, P^{[i]}_2, ..., P^{[i]}_{N_i} \right\}
    \end{equation}
    where $\bigcup_{i = 1}^{m} G^{[i]} = \{ P_i \}_{i = 1}^N$ (covering) and $[P^{[i]}_j, P^{[i]}_k] = 0$ for all $j, k = 1, ..., N_i$ (commuting). 
\end{definition}

\noindent As discussed, in an overlapped grouping the groups $G^{[i]}$ need not be disjoint, unlike in standard groupings. The idea behind this is that including Hamiltonian terms in multiple groups increases the effective number of measurement shots contributing to each term and can reduce the number of measurements required to estimate $\langle H \rangle$ to a desired accuracy, or equivalently reduce the variance of $\langle H \rangle$ given a fixed measurement budget.

We remark that more fine-grained measures of commutativity have recently been introduced to reduce measurement complexity in grouping-based algorithms~\cite{dalfaveroMeasurementReductionExpectation2025}. While we focus on (full) commutativity in this work, our methods for finding groupings and overlapped groupings can be used with these notions of commutativity as well.

\section{Methods} \label{sec:methods}

\subsection{Empirical estimators for overlapped groupings} \label{sec:empirical-estimators-for-overlapped-groupings}
Given an overlapped grouping $\mathcal{G}$, expectation values of Pauli operators can be estimated by combining the empirical estimates obtained from each group in which a Pauli appears.  This defines a family of estimators parameterized by weights $w$
\begin{equation}
    \label{eq:generalPauliEstimator}
    \overline{\langle P_i \rangle}_\mathcal{G}(w) = \sum_{j \in \Gamma_\mathcal{G}(i)} w_{i,j}\, \overline{\langle P_i \rangle}_{(j)},
\end{equation}
where $\Gamma_\mathcal{G}(i) \coloneq \left\{j:P_i \in G^{[j]}\right\}$ and the weights satisfy the unbiasedness constraint
\begin{equation}
    \label{eqn:unbiasedness-constraint-energy-estimator-weights}
    \sum_{j \in \Gamma_\mathcal{G}(i)} w_{i,j} = 1.
\end{equation}
This induces a corresponding family of energy estimators parameterized by the weights $w$
\begin{equation}
    \label{eq:generalEnergyEstimator}
    \overline{E}_\mathcal{G}(w)
    =
    \sum_i c_i \overline{\langle P_i \rangle}_\mathcal{G}(w).
\end{equation}
This flexibility in the choice of weights will play an important role in characterizing the effect of overlapped groupings on estimator variance.

In practice, we do not attempt to optimize the weights $w_{i,j}$ using covariance information.  Instead, we adopt a simple shot-weighted averaging rule which depends only on the shot allocation.  This mirrors the design philosophy of grouping-based methods: low classical cost and covariance-free grouping. Specifically, we define the heuristic estimator by choosing
\begin{equation} \label{eqn:heuristic-estimator-for-weights}
    w_{i,j} = \frac{M_j}{\sum_{k \in \Gamma_\mathcal{G}(i)} M_k},
\end{equation}
and denote the resulting Pauli and energy estimators as $\overline{\langle P_i \rangle}_\mathcal{G}$ and $\overline{E}_\mathcal{G}$.
This heuristic estimator minimizes variance when covariance is ignored by maximizing the effective shot count on each Pauli operator as is shown in Lemma~\ref{lem:oracle-reduces-to-averaging} in Sec.~\ref{subsec:analyticalResults}.  We remark that this  estimator is also derived in~\cite{wuOverlappedGroupingMeasurement2023} and re-expressed in a similar form in~\cite{yenDeterministicImprovementsQuantum2023}.  

Because different groups use disjoint sets of measurement shots, the empirical estimates from different groups are independent.  Therefore, the variance in the estimator for $\langle P_i \rangle$ is given by 
\begin{equation}
\text{Var}\!\left(\overline{\langle P_i \rangle}_{\mathcal{G}}(w)\right)
=
\sum_{j\in\Gamma_\mathcal{G}(i)} w_{i,j}^2 \,
\mathrm{Var}\big(\overline{\langle P_i \rangle}_{(j)}\big).
\end{equation} 
If two operators $P_i$ and $P_k$ appear together in at least one group, the estimator induces covariance
\begin{equation}
\begin{aligned}
\mathrm{Cov}\!\left(
\overline{\langle P_i\rangle}_{\mathcal{G}}(w),
\overline{\langle P_k\rangle}_{\mathcal{G}}(w)
\right)
&=\\
\sum_{j \in \Gamma_{\mathcal{G}}(i)\cap\Gamma_{\mathcal{G}}(k)}
&w_{i,j}\, w_{k,j}\,
\mathrm{Cov}\!\big(
\overline{\langle P_i \rangle}_{(j)},
\overline{\langle P_k \rangle}_{(j)}
\big).
\end{aligned}
\end{equation}
The total variance of the energy estimator is given by 
\begin{equation}
\label{eq:variance-expr}
\begin{aligned}
\mathrm{Var}\!\left(\overline{E}_\mathcal{G}(w)\right)
&= \sum_i c_i^2\, \mathrm{Var}\!\left(\overline{\langle P_i\rangle}_\mathcal{G}(w)\right) \\
&\quad + 2 \sum_{i<k} c_i c_k\,
\mathrm{Cov}\!\left(
\overline{\langle P_i\rangle}_\mathcal{G}(w),
\overline{\langle P_k\rangle}_\mathcal{G}(w)
\right).
\end{aligned}
\end{equation}

\subsection{Optimal shot allocation for overlapped groupings} \label{sec:optimal-shot-allocation}
Grouping-based measurement schemes are typically paired with a shot allocation strategy that distributes a fixed measurement budget across groups in order to minimize an upper bound on the estimator variance.  The most common approach in grouping-based methods is state-independent allocation, in which variances are bounded under worst case assumptions allowing covariance contributions to add constructively~\cite{weckerProgressPracticalQuantum2015,rubinApplicationFermionicMarginal2018}. Under these assumptions, the variance contribution of each group admits a simple worst-case bound, and the resulting allocation can be computed analytically at negligible classical cost.  This strategy is widely used in practice due to its robustness and independence from prior knowledge of the target state.  

For a disjoint grouping $\mathcal{G} = \{G^{[j]}\}$, the standard worst-case bound takes the form~\cite{weckerProgressPracticalQuantum2015,rubinApplicationFermionicMarginal2018}
\begin{equation}
    \text{Var}\!\left(\overline{E}\right) \leq \sum_j \frac{1}{M_j}\left(\sum_{i\in G^{[j]}} \vert c_i \vert \right)^2
\end{equation}
where $M_j$ is the number of shots allocated to group $G^{[j]}$.  Minimizing this bound subject to the constraint $\sum_j M_j = M_{tot}$ under worst-case bounds on covariance contributions yields a closed-form allocation strategy
\begin{equation}
    \label{eq: allocation}
    M_j \propto \sum_{i\in G^{[j]}} \vert c_i\vert.
\end{equation}
While this allocation strategy has been used, one can alternatively assume zero covariance and obtain
\begin{equation}
    \label{eq: zero-cov-allocation}
    M_j \propto \sqrt{\sum_{i\in G^{[j]}} c_i^2}.
\end{equation}
Throughout the theory and numerics for this work, we assume the latter allocation strategy where appropriate.

Overlapped groupings, in particular those found by repacking, modify the structure of the estimator by allowing individual operators to appear in multiple groups.  As a result, the effective variance depends not only on the allocation of shots across groups, but also on how those groups overlap.  Unlike disjoint groupings, we are not aware of a simple closed-form allocation rule for overlapped groupings.  Instead, the allocation problem naturally takes the form of a smooth constrained optimization over the group shot counts $\{M_j\}$,

\begin{equation}
    \label{eq: optallocation}
    \min_{\{M_j>0\}} \text{Var}(\overline{E}_\mathcal{G}) \quad \text{s.t. } \sum_j M_j = M_\text{tot}
\end{equation}
In the absence of state-specific information, this optimization can be carried out using the same worst-case variance bounds employed in disjoint grouping methods, yielding a state-independent allocation that is directly comparable to the strategy for grouping in~\eqref{eq: allocation}.  An alternative strategy motivated by the monotonicity result in Theorem~\ref{thm:repacking-never-increases-variance-with-covariance}, is to utilize the same allocation as implied by the initial grouping.  

Alternatively, the same framework admits state-dependent allocation, in which empirical estimates or analytical knowledge of variances and covariances are substituted into the variance model.  In this case, it becomes necessary to consider other optimization strategies due to poor scaling. Such methods are well described in prior work~\cite{yenDeterministicImprovementsQuantum2023, shlosbergAdaptiveEstimationQuantum2023}. This flexibility is not unique to overlapped groupings, as disjoint grouping schemes can also be tuned in a state-dependent manner. 

From this perspective, overlapped grouping does not fundamentally alter the role of shot allocation, but it does expand the design space over which allocation must be optimized.  Grouping heuristics determine which observables are estimable from each circuit, while allocation determines how measurement effort is distributed across those circuits.  These two choices are logically distinct but practically coupled, and overlapped grouping primarily acts by enlarging the set of estimators within a broader allocation problem.

\subsection{Repacking} \label{sec:methods-repacking}

Here we introduce two algorithms to map a (disjoint) grouping to an overlapped grouping. The algorithms produce overlapped groupings that do not form any additional groups, and only add operators to existing groups. We call such an overlapped grouping a repacked grouping, and we call a procedure which produces a repacked grouping a repacking, as defined below.

\begin{definition}[Repacked grouping, Repacking] \label{def:repacking}
Let $\G = \left\{ G^{[1]}, ..., G^{[m]} \right\}$ be a grouping.  A \textbf{repacked grouping} of $\G$ is an overlapped grouping $\mathcal{R} = \{G'^{[1]}, \ldots, G'^{[m]}\}$ such that $G^{[i]} \subseteq G'^{[i]}$ for all $i$. We refer to the algorithmic process of forming a repacked grouping from a (disjoint) grouping as \textbf{repacking}.
\end{definition}
We emphasize that the number of groups in the repacked grouping $\mathcal{R}$ is the same as the number of groups in the original (disjoint) grouping $\G$ --- that is, repacking only adds operators to existing groups, and does not form any new groups. Since $\mathcal{R}$ is a valid overlapped grouping, repacking can only add operators to groups where they are compatible.  It never deletes strings from groups, changes the number of groups (and as such, the number of circuits), or introduces incompatibility.

While these conditions are not necessary for an overlapped grouping, imposing them provides two important advantages.  First, since $G^{[i]} \subseteq G'^{[i]}$, the estimator associated with the original grouping is always contained in the enlarged estimator space. This induces a partial order on repackings, and enables a monotonic reduction in variance (Theorem~\ref{thm:repacking-never-increases-variance-with-covariance}).  Second, these conditions allow any grouping algorithm to be naturally extended to produce overlapped groupings without asymptotically increasing the runtime. Without these conditions, the problem of finding an overlapped grouping for a Hamiltonian has a (much) larger search space than finding a grouping. While in principle this enables finding a better solution, it has to be balanced with additional runtime required to do so.

Repacking has a simple interpretation on the commutativity graph $G = (V, E)$ induced by the Hamiltonian --- i.e., the graph where each node is a Pauli and edges are added between Paulis that commute. While (disjoint) groupings enforce a disjoint clique covering of $V$, repacking relaxes this condition to produce an overlapping clique cover that extends or coincides with the original covering. The underlying graph is unchanged, and in a natural greedy implementation each unordered pair of operators is tested for commutativity at most once, exactly as in sorted insertion.  As a result, the total number of compatibility checks remains $O(|V|^2)$, and repacking retains the same asymptotic classical runtime as the original grouping algorithm up to constant factors.  Of course, as the disjoint groupings are trivial repackings, and finding disjoint groupings is NP-hard, finding an optimal repacking is also NP-Hard.  

It is natural to consider iteratively reducing variance by finding successively better overlapped groupings, starting from an initial grouping. Indeed, our repacking algorithms in the next two subsections exhibit this behavior. It is therefore useful to define the notions of refinements and proper refinements below when analyzing these algorithms.

\begin{definition}
    Given a grouping $\G = \left\{ G^{[1]}, ..., G^{[m]} \right\}$ and a repacked grouping $\mathcal{R} = \{G'^{[1]}, \ldots, G'^{[m]}\}$, we say another repacked grouping $\Rcal' = \{G''^{[1]}, \dots, G''^{[m]}\}$ is a \textbf{refinement} of $\Rcal$ if $G'^{[j]} \subseteq G''^{[j]}$ for each $j=1,\dots,m$, and we say it is a \textbf{proper refinement} if additionally $G'^{[j]} \neq G''^{[j]}$ for at least one $j$. 
\end{definition}

We also define the following terminal condition of repacking in which no further refinement is possible.

\begin{definition}
    A \textbf{maximally repacked grouping} of $\Gcal$ is a repacked grouping that cannot give rise to a proper refinement.
\end{definition}

\begin{theorem}
\label{prop:maximal-repacking-non-uniqueness}
Let $\Gcal =  \{G^{[1]}, \ldots, G^{[m]}\}$ be a grouping, and $\Rcal =  \{G'^{[1]}, \ldots, G'^{[m]}\}$ be a repacked grouping of $\Gcal$. Then there exists a maximally repacked grouping of $\Gcal$, such that it is a refinement of $\Rcal$. Moreover, $\Gcal$ can have multiple distinct maximally repacked groupings.
\end{theorem}

\begin{proof}
Suppose the Hamiltonian under consideration is $H = \sum_{i=1}^{N} c_i P_i$. Starting with $\Rcal$, let us create a refinement $\Rcal_1$ of $\Rcal$ by taking the Pauli operator $P_1$ and adding it to each $G'^{[j]} \in \Rcal$ such that $G'^{[j]} \cup P_1$ is still a commuting set. Next, we repeat the  same process for $P_2$, and create a refinement $\Rcal_2$ of $\Rcal_1$. This process is continued until we obtain a final refinement $\Rcal_N$ of $\Rcal_{N-1}$. Then $\Rcal_N$ is clearly a maximally repacked grouping by construction, and it is a refinement of $\Rcal$.

For the non-uniqueness part of maximally repacked groupings, we furnish an explicit example. Consider $\Gcal = \{\{IZI,YZX\},\{IXX,IYY\},\{XII\}\}$.  Two possible repacked groupings of $\Gcal$ are given by $\Rcal = \{\{IZI,YZX\},\{IXX,IYY\},\{XII,IZI\}\}$, and $\Rcal' = \{\{IZI,YZX\},\{IXX,IYY\},\{XII,IXX,IYY\}\}$. Note that there is no common refinement of $\Rcal$ and $\Rcal'$, since $IZI$ and $IYY$ do not commute.
\end{proof}

In the following sections, we introduce two particular repacking algorithms.

\subsubsection{Post-hoc repacking}

\begin{algorithm}
\label{alg:posthoc_repacking}
\footnotesize

Initialize $G'^{[j]} \gets G^{[j]}$ for all $j=1,\dots,m$\;

\For{$j=1,\dots,m$}{
    Let $U_j$ be the diagonalizing unitary for $G^{[j]}$\;

    \For{each $P_i \in H$}{
        $P_i' \gets U_j P_i U_j^\dagger$\;

        \If{$P_i'$ contains only $Z$ and $I$}{
            $G'^{[j]} \gets G'^{[j]} \cup \{P_i\}$\;
        }
    }
}

Output $\mathcal{R} = \{G'^{[1]},\dots,G'^{[m]}\}$\;
\caption{Post-hoc repacking}
\end{algorithm}

The post-hoc repacking scheme is motivated by basis-first measurement schemes such as classical shadows \cite{huangPredictingManyProperties2020} and derandomized shadows \cite{huangEfficientEstimationPauli2021}.  In those approaches, one first samples a measurement basis, most commonly a tensor-product of single-qubit Pauli operators, and then identifies all Pauli observables that are diagonal in that basis.  Any such observable can be estimated using the same set of bitstrings.  In this sense, the basis determines the grouping.

Grouping methods adopt the opposite perspective: a group of commuting target observables is first selected, and then a Clifford unitary $U_g$ is applied to diagonalize that entire group simultaneously~\cite{aaronsonImprovedSimulationStabilizer2004,bravyiCliffordCircuitOptimization2021}. However, unless the target group is complete, applying $U_g$ necessarily diagonalizes additional Pauli strings beyond the ones that were explicitly grouped, some of which may belong to $H$ and thus could contribute to the energy estimate if their expectation values were extracted.  

Post-hoc repacking exploits this observation.  For each group with fixed diagonalizing unitary $U_g$, we consider every Pauli operator $P_i \in H$ and compute the conjugated operator $P_i' = U_g P_i U_g^\dagger$.  If $P_i'$ contains only $Z$ and $I$ factors, then $P_i$ is diagonal in the measurement basis associated with $U_g$.  In that case, the bitstrings already collected from measuring group $G$ also provide information for estimating $\langle P_i \rangle$.
Computing this repacking requires $O(m|V|)$ conjugations where $m$ is the number of groups, which in the worst case matches the $O(|V|^2)$ scaling of the commutativity checks used to form the initial grouping and thus does not increase the asymptotic classical cost.

This procedure enables variance reduction using only previously collected data --- no additional quantum circuits are executed, no new measurement settings are introduced, and all improvements arise purely from post-processing the original bitstrings. As such, post-hoc repacking strictly generalizes any grouping-based experiment at the estimator level, enabling the experimenter to reuse all past data more effectively and to recover expectation values for additional terms that were implicitly diagonalized but not originally included in the measurement groups.  

Post-hoc repacking exposes a subtle but important mismatch between the estimator and the information content of the measurement data.  In a typical grouping-based experiment, the recorded bitstrings already contain unbiased information about operators that are implicitly diagonal in the applied measurement basis, regardless of whether those operators were explicitly targeted by the grouping procedure.  For such untargeted observables, the boundary between measured and unmeasured becomes a property of the estimator rather than of the data itself.  If an unbiased estimator exists that incorporates this additional information without increasing experimental cost, then continuing to use an estimator which systematically discards it is no longer a neutral modeling choice.  Post-hoc repacking therefore serves as a principled way of aligning the estimator with the full information content already present in the measurement data.  This observation does not imply that every practical estimator must exploit all such information, but it clarifies that any decision not to do so should be justified by an explicit tradeoff.  

With this in mind, a proactive choice can be made to intentionally repack groups to reduce variance or other cost metrics through an ad-hoc procedure based on continuing the greedy algorithm from the original grouping.

\subsubsection{Ad-hoc repacking}

\begin{algorithm}
\caption{Ad-hoc repacking}
\label{alg:adhoc-repacking}
\footnotesize

Initialize $G'^{[j]} \gets G^{[j]}$ for all $j=1,\dots,m$\;

\For{each Pauli operator $P_i$}{
    $\mu_i \gets \left|\{j : P_i \in G'^{[j]}\}\right|$\;
}

\While{there exists $(P_i,j)$ such that $P_i \notin G'^{[j]}$ and $P_i$ commutes with all operators in $G'^{[j]}$}{

    Select $P_i$ maximizing $C(P_i) = c_i^2 / \mu_i$\;

    Insert $P_i$ into the first compatible group $G'^{[j]}$ such that $P_i \notin G'^{[j]}$\;

    $G'^{[j]} \gets G'^{[j]} \cup \{P_i\}$\;

    $\mu_i \gets \mu_i + 1$\;
}

Output $\mathcal{R} = \{G'^{[1]},\dots,G'^{[m]}\}$\;

\end{algorithm}

The ad-hoc repacking procedure enables further variance reduction by modifying the grouping before constructing the diagonalizing circuits.  Conceptually, ad-hoc repacking occupies a middle ground between the basis-first and grouping-first perspectives.  A fixed basis uniquely determines which Pauli operators are jointly diagonal, and thus determines a grouping, but a fixed grouping generally admits many valid diagonalizing bases.  Ad-hoc repacking exploits this asymmetry.  Rather than deriving groups from a basis, it selectively enlarges the groups themselves so that the subsequent basis construction yields richer diagonal structure.  

We extend each original group by inserting additional compatible operators, using the same sorted insertion rule as in the original grouping procedure.  That is, each Pauli string is placed into the earliest existing group for which it it is compatible and of which it is not already a member.  To heuristically reduce variance, we score each string by its independent shot-noise contribution to the estimated variance while ignoring covariance under a simple allocation strategy, and insert strings greedily into groups in descending score until all strings are exhausted.  The score for each string is chosen to be
\begin{equation}
    C(P_i) = c_i^2/\mu_i 
\end{equation}
where $c_i$ is the coefficient of $P_i$ in $H$ and $\mu_i$ is the number of groups in which $P_i$ appears.  After each insertion of $P_i$, the score $C(P_i)$ is recomputed.  This scoring can be viewed as a natural extension of the sorted insertion algorithm to the repacking setting, where multiplicity across groups down-weights the priority of each term. Pseudocode for the ad-hoc repacking algorithm is shown in Algorithm~\ref{alg:adhoc-repacking}. Notice that the output of Algorithm~\ref{alg:adhoc-repacking} is always a maximally repacked grouping.

A related greedy extension of disjoint groupings to overlapped groupings was previously proposed in~\cite{yenDeterministicImprovementsQuantum2023} where a similar coefficient-based scoring rule was stated without the down-weighting by group multiplicity used here.  Although these scoring rules are natural and extremely cheap to evaluate, they are not expected to be globally optimal, even within covariance free heuristics.  Beyond the local, suboptimal choices made by the greedy algorithm, this scoring does not take into account optimal allocation so it cannot accurately reflect the variance reduction for each repacking.  Our goal here is not to claim optimality, but to illustrate how modifying the grouping before circuit synthesis provides a tractable and effective means of variance reduction.  A more refined scoring method could incorporate partial variance and covariance estimates for a given state as in \cite{greschGuaranteedEfficientEnergy2025, yenDeterministicImprovementsQuantum2023, shlosbergAdaptiveEstimationQuantum2023}, a penalty for increasing circuit depth, or other structural features of the Hamiltonian.

\section{Results} \label{sec:results}

Here we present and prove our analytical results (Sec.~\ref{subsec:analyticalResults}) and present our numerical results (Sec.~\ref{sec:numerical-results}).

\subsection{Analytical results}
\label{subsec:analyticalResults}

The primary goal of this section is to  prove the main theorems from Sec.~\ref{sec:main-results}. From this section, Theorem~\ref{thm:maximal-variance-reduction} is restated and proved as Theorem~\ref{thm:maximal-variance-reduction-restatement}, and Theorem~\ref{thm:repacking-reduces-variance} is restated and proved as Theorem~\ref{thm:monotonicity-zero-cov}. Intermediate and additional results are stated and proved along the way.

\newpage

\subsubsection{Maximal variance reduction for overlapped grouping:\\Proof of Theorem~\ref{thm:maximal-variance-reduction} (Restated as Theorem~\ref{thm:maximal-variance-reduction-restatement})}

In this section we prove Theorem~\ref{thm:maximal-variance-reduction} on the maximal variance reduction possible for overlapped grouping, which we restate for convenience below. Throughout this section, we refer to Fig.~\ref{fig:partial_order} to illustrate the family of Hamiltonians which admit this maximal variance reduction, as well as the particular groupings $\G$ and $\G'$ and the repacked overlapped grouping $\Rcal$ from this figure. 

\begin{theorem} (Restatement of Theorem~\ref{thm:maximal-variance-reduction}.) 
\label{thm:maximal-variance-reduction-restatement}
There exist Hamiltonians and states for which the variance reduction of overlapped grouping relative to the best disjoint grouping found by sorted insertion is $\Theta(m)$, where $m$ is the number of disjoint groups. Assuming a model of state-independent variance and zero covariance, this is the asymptotic maximal variance reduction that can be achieved by overlapped grouping. 
\end{theorem}

At a high level, the proof sketch is as follows. We show that even when shot allocation is chosen optimally under the standard state-independent variance model, assuming a maximally mixed state and vanishing covariance, the sorted insertion grouping heuristic can exhibit arbitrarily worse variance under repacking.  This is demonstrated by identifying two different valid groupings of the same Hamiltonian, one of which arises from sorted insertion, and showing that the variance of sorted insertion is worse by factor of $\Theta(m)$.  Then, we show that both groups can be repacked into the same grouping as in Fig.~\ref{fig:partial_order}. By invoking the partial order, we then see that the repacked grouping achieves a $\Theta(m)$ reduction in variance over sorted insertion.

First, we show that the variance estimator~\eqref{eqn:heuristic-estimator-for-weights} is optimal in the limit of zero covariance, in the sense that it provides the minimum-variance unbiased energy estimator.

\begin{lemma}
\label{lem:oracle-reduces-to-averaging}
Fix a grouping $\mathcal{G}=\{G_1,\dots,G_m\}$ and shot counts $\{M_j\}_{j=1}^m$.
Assume:
\begin{enumerate}
    \item The shot batches across different groups are independent.
    \item For any group $G^{[j]}$, the estimators of distinct Pauli terms measured in that group have zero covariance, i.e.,
\[
\mathrm{Cov}\!\big(
\overline{\langle P_i\rangle}_{(j)},
\overline{\langle P_k\rangle}_{(j)}
\big) = 0
\quad \text{for all distinct } P_i, P_k \in G^{[j]}.
\]
\end{enumerate}
Then the minimum-variance unbiased energy estimator~\eqref{eq:generalEnergyEstimator}
is attained by the heuristic estimator~\eqref{eqn:heuristic-estimator-for-weights}.
\end{lemma}

\begin{proof}
Under assumptions (1)--(2), all covariance terms between distinct Pauli observables vanish. Therefore, the energy variance~\eqref{eq:variance-expr} becomes
\begin{equation}
\mathrm{Var}\!\left(\overline{E}_\mathcal{G}(w)\right)
=
\sum_i c_i^2\,\mathrm{Var}\!\left(\overline{\langle P_i\rangle}_\mathcal{G}(w)\right).
\end{equation}
For each $i$, we have
\begin{equation}
\overline{\langle P_i\rangle}_\mathcal{G}(w)
=
\sum_{j\in \Gamma_\mathcal{G}(i)} w_{i,j}\,\overline{\langle P_i\rangle}_{(j)}.
\end{equation}
Since the shot batches across different groups are independent
\begin{equation}
\mathrm{Var}\!\left(\overline{\langle P_i\rangle}_\mathcal{G}(w)\right)
=
\sum_{j\in \Gamma_\mathcal{G}(i)} w_{i,j}^2\,
\mathrm{Var}\!\big(\overline{\langle P_i\rangle}_{(j)}\big).
\end{equation}

For fixed $i$, minimizing the energy variance reduces to the constrained convex quadratic program
\begin{equation}
\min_{\{w_{i,j}\}}
\sum_{j\in \Gamma_\mathcal{G}(i)} w_{i,j}^2\,
\mathrm{Var}\!\big(\overline{\langle P_i\rangle}_{(j)}\big)
\quad\text{s.t.}\quad
\sum_{j\in \Gamma_\mathcal{G}(i)} w_{i,j}=1.
\end{equation}
Introducing a Lagrange multiplier $\lambda_i$ and differentiating yields
\begin{equation}
2 w_{i,j}\,\mathrm{Var}\!\big(\overline{\langle P_i\rangle}_{(j)}\big)
- \lambda_i = 0
\quad\Rightarrow\quad
w_{i,j} \propto \frac{1}{\mathrm{Var}\!\big(\overline{\langle P_i\rangle}_{(j)}\big)}.
\end{equation}
When $P_i$ is measured $M_j$ times in group $j$,
\begin{equation}
\mathrm{Var}\!\big(\overline{\langle P_i\rangle}_{(j)}\big)
=
\frac{1 - \langle P_i\rangle^2}{M_j},
\end{equation}
which differs across $j$ only through the factor $1/M_j$. Therefore,
\begin{equation}
w_{i,j}^\star \propto M_j,
\end{equation}
and normalizing gives
\begin{equation}
w_{i,j}^\star
=
\frac{M_j}{\sum_{k\in \Gamma_\mathcal{G}(i)} M_k},
\end{equation}
which is exactly~\eqref{eqn:heuristic-estimator-for-weights}.
\end{proof}

Next, we characterize the variance of the groupings.

\begin{lemma}
\label{prop:si-suboptimal}
There exist Hamiltonians for which sorted insertion   produces a grouping whose estimator variance, under optimal state-independent shot allocation, exceeds that of another valid grouping by a factor of $\Theta(m)$, where $m$ is the number of groups.
\end{lemma}

\begin{proof}
Let $\mathcal{L} = \{a,b,c,\dots\}$ be a set of $L$ lowercase labels and let $\mathcal{U} = \{A,B,C,\dots\}$ be a corresponding set of uppercase labels of the same cardinality.  
Construct a Hamiltonian $H=\sum_i c_i P_i$ containing Pauli operators labeled by ordered pairs of letters drawn from $\mathcal{U}\times\mathcal{L}$, together with lowercase-only operators.

Specifically, for each lowercase letter $\ell\in\mathcal{L}$, include the operators
$
A\ell,\; B\ell,\; C\ell,\;\dots
$
(one for each uppercase letter), and additionally include the lowercase-only operator $\ell$ for all but one choice of $\ell$. Assume the commutativity relations satisfy:
\begin{enumerate}
    \item All operators whose labels contain an uppercase letter mutually commute.
    \item A lowercase-only operator $\ell$ commutes with exactly those operators whose label ends in the same lowercase letter $\ell$, and anticommutes with all operators whose lowercase label differs.
\end{enumerate}
Refer to Fig.~\ref{fig:partial_order} for an illustration of the Hamiltonian (commutativity graph).

Choose coefficients $c_i$ such that $|c_i|= 1 + \epsilon = O(1)$ for all operators. Here $-1 / N \le \epsilon \le 1 / N$ is a small perturbation and $N$ is the number of nodes in the graph (terms in the Hamiltonian).
Under sorted insertion with suitable choices of perturbation for the coefficients, this construction yields $m=L$ groups of the form
\begin{equation}
\begin{aligned}
\mathcal{G}
&=
\Bigl\{
\{Aa,Ba,Ca,\dots,a\},
\{Ab,Bb,Cb,\dots,b\}, \\
&\qquad
\dots,\,
\{Az,Bz,Cz,\dots\}
\Bigr\},
\end{aligned}
\end{equation}
where the final group contains no lowercase-only operator. Now consider the alternative valid grouping
\begin{equation}
\begin{aligned}
\mathcal{G}' = {}&
\Bigl\{
\{Aa,Ba,\dots,Ab,Bb,\dots,\dots,Az,Bz,\dots\}, \\
&\qquad
\{a\},\{b\},\dots
\Bigr\}
\end{aligned}
\end{equation}
consisting of a single group containing all operators with uppercase letters and singleton groups for each lowercase-only operator. Refer to Fig.~\ref{fig:partial_order} for illustrations of $\G$ and $\G'$.

Assume the standard state-independent variance model used for shot allocation and a worst-case maximally mixed state.  This implies each Pauli observable has variance $1$ and that covariance terms vanish.  
For a disjoint grouping $\mathcal{G}=\{G^{[j]}\}$ with shot allocation $\{M_j\}$ and total budget $\sum_j M_j=M$, the estimator variance is
\begin{equation}
\mathrm{Var}(\overline{E})
=
\sum_j \frac{S_j}{M_j}
\end{equation}
where $S_j := \sum_{i\in G^{[j]}} c_i^2$. This is minimized by allocating $M_j \propto \sqrt{S_j}$, yielding
\begin{equation}
\min \mathrm{Var}(\overline{E})
=
\frac{\bigl(\sum_j \sqrt{S_j}\bigr)^2}{M}.
\end{equation}
Note that under the assumptions of this theorem and Lemma~\ref{lem:oracle-reduces-to-averaging}, the heuristic estimator $\overline{E}_\mathcal{G}$ coincides with this estimator.

In $\mathcal{G}'$, the large group contains $L^2$ operators with $c_i^2=\Theta(1)$, so $S_{\mathrm{big}}=\Theta(L^2)$ and $\sqrt{S_{\mathrm{big}}}=\Theta(L)$, while each singleton group contributes $O(1)$. Hence,
\begin{equation}
\min \mathrm{Var}\!\left(\overline{E}_{\mathcal{G}'}\right)
=
\Theta\!\left(\frac{L^2}{M}\right).
\end{equation}

In $\mathcal{G}$, each of the first $L-1$ groups contains $L$ uppercase-labeled operators and one lowercase-only operator, giving $S_j=\Theta(L+1)$, while the final group has $S=\Theta(L)$. Therefore,
\begin{equation}
\min \mathrm{Var}(\overline{E}_\mathcal{G})
=
\Theta\!\left(\frac{L^3}{M}\right).
\end{equation}
Taking the ratio yields
\begin{equation}
\frac{\min \mathrm{Var}(\overline{E}_\mathcal{G})}
{\min \mathrm{Var}(\overline{E}_\mathcal{G'})}
=
\Theta(L)
=
\Theta(m),
\end{equation}
which proves the claim.
\end{proof}

Next, we establish that the overlapped grouping $\Rcal$ shown in Fig.~\ref{fig:partial_order} is a repacking of the (disjoint) groupings $\G$ and $\G'$ shown in this same figure and discussed in the previous proof.

\begin{lemma}[A common repacking dominates both $\mathcal{G}$ and $\mathcal{G}'$]
\label{lem:common-repacking}
In the same setting as Theorem~\ref{thm:maximal-variance-reduction-restatement}, there exists an overlapped $\mathcal{R}$ such that
(i) $\mathcal{R}$ is a repacking of the sorted insertion grouping $\mathcal{G}$, and
(ii) $\mathcal{R}$ is also a repacking of the alternative grouping $\mathcal{G}'$.
\end{lemma}

\begin{proof}
Define the grouping $\mathcal{R}$ on the same $m=L$ group indices as follows. For each lowercase label
$\ell\in\mathcal{L}$ other than the distinguished letter $z$ (the one for which no lowercase-only term is included),
let
\[
R^{[\ell]} := \{A\ell, B\ell, C\ell, \dots, \ell\},
\]
and 
\[
R^{[z]} := \{A a, B a, \dots, A b, B b, \dots, \dots, A z, B z, \dots\},
\]
containing all operators whose labels contain an uppercase letter. Refer to Fig.~\ref{fig:partial_order} for an illustration of $\Rcal$.

We first show that $\mathcal{R}$ is a repacking of $\mathcal{G}$. By construction,
$\mathcal{G}$ already contains $R^{[\ell]}$ for each $\ell\neq z$, and its final group contains the subset
$\{A z, B z, \dots\} \subseteq R^{[z]}$. Since all uppercase-labeled operators mutually commute, we may insert the
remaining uppercase-labeled operators into the final sorted insertion group without creating incompatibility. Hence each original sorted insertion
group is a subset of the corresponding $R^{[\ell]}$, and no new groups are created, so $\mathcal{R}$ is a valid repacking
of $\mathcal{G}$.

Next, we show that $\mathcal{R}$ is a repacking of $\mathcal{G}'$. The grouping $\mathcal{G}'$ contains the large
group of all uppercase-labeled operators, which is exactly $R^{[z]}$. For each singleton group $\{\ell\}$ with $\ell\neq z$,
we may insert the operators $\{A\ell,B\ell,C\ell,\dots\}$, because the construction assumes that $\ell$ commutes with
exactly those uppercase-labeled operators carrying the same lowercase index $\ell$. Thus $\{\ell\}\subseteq R^{[\ell]}$ and
compatibility is preserved. Again, no new groups are created, so $\mathcal{R}$ is a repacking of $\mathcal{G}'$.
\end{proof}

In the following section, we show in Theorem~\ref{thm:repacking-never-increases-variance-with-covariance} that repacking never increases variance, in particular that $\mathrm{Var}^*(\Rcal) \le \mathrm{Var}^*(G)$ for any repacking $\Rcal$ of grouping $\G$ where
\begin{equation}
    \text{Var}^*(\Rcal) = \min_w \text{Var}(\overline{E}_\Rcal(w))
\end{equation}
and similarly for $\mathrm{Var}^*(\mathcal{G})$.
From Lemma~\ref{lem:common-repacking}, we thus have that $\mathrm{Var}^*(\mathcal{R}) \le \mathrm{Var}^*(\mathcal{G})$
and $\mathrm{Var}^*(\mathcal{R}) \le \mathrm{Var}^*(\mathcal{G}')$. 
By Lemma~\ref{prop:si-suboptimal}, it follows that
\[
\frac{\mathrm{Var}^*(\mathcal{G})}{\mathrm{Var}^*(\mathcal{R})}
\ge
\frac{\mathrm{Var}^*(\mathcal{G})}{\mathrm{Var}^*(\mathcal{G}')}
=
\Theta(m).
\]
This bound is asymptotically tight under the state-independent variance model with vanishing covariance and optimal shot allocation. 
In this setting, the estimator variance is given by
\begin{equation}
\mathrm{Var}^*\!\left(\overline{E}_\mathcal{G}\right)
= \frac{\left(\sum_j \sqrt{S_j}\right)^2}{M}
\end{equation}
where $S_j = \sum_{i \in G^{[j]}} c_i^2$. 
By concavity of the square root, this quantity is minimized when all operators are grouped together and maximized when weight is distributed evenly across groups. 
Consequently, the maximal variance reduction achievable by overlapped grouping is $\Theta(m)$, which completes the proof of Theorem~\ref{thm:maximal-variance-reduction-restatement}. 

We remark that by changing the allocation strategy from state-independent optimal allocation to uniform allocation and changing the initial grouping to $\mathcal{G'}$, the same commutativity structure yields an asymptotically identical variance reduction through repacking to the same grouping $\mathcal{R}$, but with  $\mathcal{G}$ playing the role of the optimal grouping.  Although the optimal allocation is superior in absolute terms, this role reversal highlights that repacking can endow a grouping with a degree of  robustness to suboptimal allocation strategies.  Consistent with this interpretation, our numerical results in Sec.~\ref{sec:numerical-results} show that repacking offers its largest variance reductions whenever the initial grouping or allocation strategy is poorly matched to the Hamiltonian or grouping, while yielding more modest reductions whenever these choices are already well aligned with the problem.  While in application it is prudent to use quality allocation strategies, this effect demonstrates that overlapping groupings provide support for the use of approximate or adaptive allocation strategies which may not realize full alignment with the optimal strategy.

\subsubsection{Deterministic algorithms for reducing variance:\\Proof of Theorem~\ref{thm:repacking-reduces-variance} (Restated as Theorem~\ref{thm:monotonicity-zero-cov})}

The primary goal of this section is to prove Theorem~\ref{thm:repacking-reduces-variance}, which we restate for convenience below. Intuitively, this result says that overlapped groupings found by repacking always have smaller variance than the disjoint groupings from which they arise, assuming zero covariance and a mild assumption on the repacking step.

\begin{theorem} (Restatement of Theorem~\ref{thm:repacking-reduces-variance}.)
\label{thm:monotonicity-zero-cov}
Let $\mathcal{G}$ be a grouping of a Hamiltonian $H$ and let $\mathcal{R}$ be a repacking of $\mathcal{G}$.  Assume that for all distinct Pauli operators $P_i, P_k$ in $\mathrm{supp}(H)$, $\sigma_{P_i P_k} = 0.$
Fix a shot allocation $\{M_j\}_{j=1}^m$ over the groups such that group $G^{[j]}$ and its corresponding repacked group $G'^{[j]}$ are measured $M_j$ times.  Using a shot-weighted averaging for the energy estimators $\ovl{E}_{\Gcal}$ and $\ovl{E}_{\Rcal}$, we have that 
\begin{equation}
    \Vr{\ovl{E}_{\Rcal}} < \Vr{\ovl{E}_{\Gcal}}
\end{equation}
if at least one refinement step adds a Pauli operator $P_s$ with $\sigma^2_{P_s} > 0$.
\end{theorem}

In addition to this main result, we also prove additional results. First, we show that the variance of an overlapped grouping found by repacking never increases the variance of the initial disjoint grouping. This result does not make any assumptions about covariance, and so can be thought of as a generalized Theorem~\ref{thm:monotonicity-zero-cov}.

\begin{theorem}
\label{thm:repacking-never-increases-variance-with-covariance}
Let $\mathcal{G}$ be a grouping of a Hamiltonian $H$ and let $\mathcal{R}$ be a repacking of $\mathcal{G}$.  Fix a shot allocation $\{M_j\}_{j=1}^m$ over the groups such that group $G^{[j]}$ and its corresponding repacked group $G'^{[j]}$ are measured $M_j$ times.  Assume that for any choice of weights $w$ the variances $\text{Var}(\overline{E}_\mathcal{G}(w))$ and $\text{Var}(\overline{E}_\mathcal{R}(w))$ (Eqn.~\eqref{eq:variance-expr}) can be evaluated exactly, for example because the true within-group covariance matrix of all measured Pauli outcomes is known.  Define
\begin{equation}
    \text{Var}^*(\mathcal{G}) = \min_w \text{Var}(\overline{E}_\mathcal{G}(w))
\end{equation}
and 
\begin{equation}
    \text{Var}^*(\mathcal{R}) = \min_w \text{Var}(\overline{E}_\mathcal{R}(w))
\end{equation}
Then,
\begin{equation}
    \text{Var}^*(\mathcal{R}) \leq \text{Var}^*(\mathcal{G}).
\end{equation}
\end{theorem}

\begin{proof}
    Because $G^{[i]} \subseteq G'^{[i]} \quad \forall i$, $\Gamma_\mathcal{G}(i) \subseteq \Gamma_\mathcal{R}(i)$.  Take any unbiased estimator $\overline{E}_\mathcal{G}(w)$ defined by weights $w_{i,j}$ with $j \in \Gamma_{\mathcal{G}}(i)$.  Define extended weights
    \begin{equation}
        \tilde w_{i,j} =
\begin{cases}
w_{i,j}, & j \in \Gamma_{\mathcal{G}}(i),\\
0,       & j \in \Gamma_{\mathcal{R}}(i)\setminus \Gamma_{\mathcal{G}}(i).
\end{cases}
    \end{equation}
The weights still satisfy the unbiasedness constraint~\eqref{eqn:unbiasedness-constraint-energy-estimator-weights} so $\overline{E}_{\mathcal{R}}(\tilde{w})$ is a valid unbiased estimator~\eqref{eq:generalEnergyEstimator}.  This estimator has identical variance to the original grouping as all repacked terms have no weight.  Therefore, every estimator achievable under $\mathcal{G}$ is achievable under $\mathcal{R}$.  Minimizing over a superset cannot increase the minimum, therefore $\text{Var}^*(\mathcal{R}) \leq \text{Var}^*(\mathcal{G})$
\end{proof}

We remark that this result endows overlapped groupings with a partial order structure, and was used to complete the proof of Theorem~\ref{thm:maximal-variance-reduction-restatement} in the previous subsection.

Now, we find necessary and sufficient conditions that allows refinement of a repacked grouping under the assumption of shot-weighted averaging scheme, by changing any of the existing groups by adding a single Pauli operator. Again let $\Gcal$ be a grouping and $\Rcal= \{G^{[1]},\dots,G^{[m]}\}$ is a repacked grouping of $\Gcal$. Suppose the measurement budgets allocated to group $G^{[j]}$ is given by a positive integer $M_j$, and we assume $M_1 + \dots + M_m = M$. Then the expression for the variance of the estimated energy $\ovl{E}_\Rcal$ from Eq.~\eqref{eq:variance-expr}, under the shot-weighted averaging scheme becomes
\begin{equation}
\label{eq:var-formula-repack-shot-wt}
\begin{split}
& \Vr{\ovl{E}_\Rcal} = \sum_{i=1}^{N} \frac{c_i^2 \sigma^2_{P_i}}{ \left( \sum_{j \in \Gamma_{\Rcal}(i)} M_j\right)} \\
&+ 2 \sum_{\substack{i,k=1 \\ i < k}}^{N}  \frac{c_i c_k \sigma_{P_i P_k} \left( \sum_{j \in \Gamma_{\Rcal}(i) \cap \Gamma_{\Rcal}(k)} M_j\right) }{\left( \sum_{j \in \Gamma_{\Rcal}(i)} M_j\right) \left( \sum_{j \in \Gamma_{\Rcal}(k)} M_j\right)}.
\end{split}
\end{equation}
We can then prove the following.

\begin{lemma}[One-step refinement]
\label{lem:variance-change-one-step}
Suppose $\Rcal' = \{G'^{[1]},\dots,G'^{[m]}\}$ is a refinement of $\Rcal$. Assume that there exists $1 \le \ell \le m$ such that (i) $G'^{[j]} = G^{[j]}$ for all $j \ne \ell$, and (ii) $G'^{[\ell]} \setminus G^{[\ell]} = \{P_s\}$, for some $P_s \in \text{supp}(H)$. Define the integer-valued quantities
\begin{equation}
\begin{split}
    \alpha_t = \sum_{j \in \Gamma_{\Rcal}(t)} M_j, \qquad \alpha_{t,t'} = \sum_{j \in \Gamma_{\Rcal}(t) \cap \Gamma_{\Rcal}(t')} M_j,
\end{split}
\end{equation}
for all $t,t' = 1,\dots,N$, and also define the sets $\mathcal{Q}_\ell = \{j : P_j \in G^{[\ell]}\}$ and $\mathcal{Q}_\ell^\complement = \{1,\dots,N\} \setminus \mathcal{Q}_\ell$. Then for the shot-weighted averaging scheme, we have $\Vr{\ovl{E}_{\Rcal'}} ]\le \Vr{\ovl{E}_{\Rcal}}$, when each group is measured with the same measurement budgets $M_1, \dots, M_m$, if and only if
\begin{equation}
\label{eq:variance-change-one-step-2}
\begin{split}
\frac{1}{2} c_s^2 \sigma^2_{P_s}  & \ge  \sum_{i \in \mathcal{Q}_\ell} c_i c_s \sigma_{P_i P_s} \left( \frac{\alpha_s - \alpha_{i,s}}{\alpha_i} \right) \\
& - \sum_{s \neq i \in \mathcal{Q}^\complement_\ell} c_i c_s \sigma_{P_i P_s}   \left( \frac{\alpha_{i,s}}{\alpha_i } \right).
\end{split}
\end{equation}
We have $\Vr{\ovl{E}_{\Rcal'}} < \Vr{\ovl{E}_{\Rcal}}$ if and only if the inequality \eqref{eq:variance-change-one-step-2} is strict.
\end{lemma}

\begin{proof}
Let us define $\delta_{i,\ell} = 1$ if $P_i \in G^{[\ell]}$, otherwise $0$, for all $i=1,\dots,N$. Then using Eq.~\eqref{eq:var-formula-repack-shot-wt} we have
\begin{equation}
\label{eq:variance-change-one-step-1}
\begin{split}
&\Vr{\ovl{E}_{\Rcal}} - \Vr{\ovl{E}_{\Rcal'}} = \frac{c_s^2 \sigma^2_{P_s} M_\ell}{\alpha_s (\alpha_s + M_\ell)} \\
&+ 2 \sum_{\substack{i=1 \\ i \neq s}}^{N} c_i c_s \sigma_{P_i P_s} \left( \frac{\alpha_{i,s}}{\alpha_i \alpha_s} - \frac{\alpha_{i,s} + M_\ell \delta_{i,\ell}}{\alpha_i (\alpha_s + M_\ell)} \right).
\end{split}
\end{equation}
Breaking up the second term on the right hand side over the sets $\mathcal{Q}_\ell$ and $\mathcal{Q}^\complement_\ell$, and then simplifying the terms yield
\begin{equation*}
\begin{split}
&\Vr{\ovl{E}_{\Rcal}} - \Vr{\ovl{E}_{\Rcal'}} = \frac{c_s^2 \sigma^2_{P_s} M_\ell}{\alpha_s (\alpha_s + M_\ell)} \\
&+ 2 \sum_{s \neq i \in \mathcal{Q}^\complement_\ell} \frac{c_i c_s \sigma_{P_i P_s}  \alpha_{i,s} M_\ell}{\alpha_i \alpha_s (\alpha_s + M_\ell)} \\
&+ 2 \sum_{i \in \mathcal{Q}_\ell} \frac{c_i c_s \sigma_{P_i P_s} M_\ell (\alpha_{i,s} - \alpha_s)}{\alpha_i \alpha_s (\alpha_s + M_\ell)}.
\end{split}
\end{equation*}
Using $\alpha_s, M_\ell > 0$, we immediately obtain from the above equation that $\Vr{\ovl{E}_{\Rcal'}} \le \Vr{\ovl{E}_{\Rcal}}$ if and only if the inequality in Eq.~\eqref{eq:variance-change-one-step-2} holds, and $\Vr{\ovl{E}_{\Rcal'}} < \Vr{\ovl{E}_{\Rcal}}$ if and only if the inequality is strict.
\end{proof}

Now, from Eqn.~\eqref{eq:variance-change-one-step-1} in the proof of Lemma~\ref{lem:variance-change-one-step}, we immediately see the following.

\begin{corollary}
\label{cor:monotonicity-shot-weighted}
In the same setting as Lemma~\ref{lem:variance-change-one-step}, suppose that we additionally have $\sigma_{P_i P_s} = 0$, for all $1 \le i \le N$ with $i \ne s$. Then, $\Vr{\ovl{E}_{\Rcal'}} \le \Vr{\ovl{E}_{\Rcal}}$, and $\Vr{\ovl{E}_{\Rcal'}} < \Vr{\ovl{E}_{\Rcal}}$ if and only if $\sigma^2_{P_s} > 0$.
\end{corollary}

The key point about Corollary~\ref{cor:monotonicity-shot-weighted} is that it holds irrespective of what fixed shot budgets $M_1, \dots, M_m$ are used. One should contrast this with Theorem~\ref{thm:repacking-never-increases-variance-with-covariance}, where monotonicity is established for optimal weights, but without the zero covariance assumption which is utilized to derive the heuristic estimator. If we assume zero covariance for all pairs of Pauli operators in $\text{supp}(H)$, by repeated application of Corollary~\ref{cor:monotonicity-shot-weighted}, we obtain Theorem~\ref{thm:monotonicity-zero-cov}.

\subsection{Numerical results} \label{sec:numerical-results}

\begin{figure*}
    \centering
    \includegraphics[width=\textwidth]{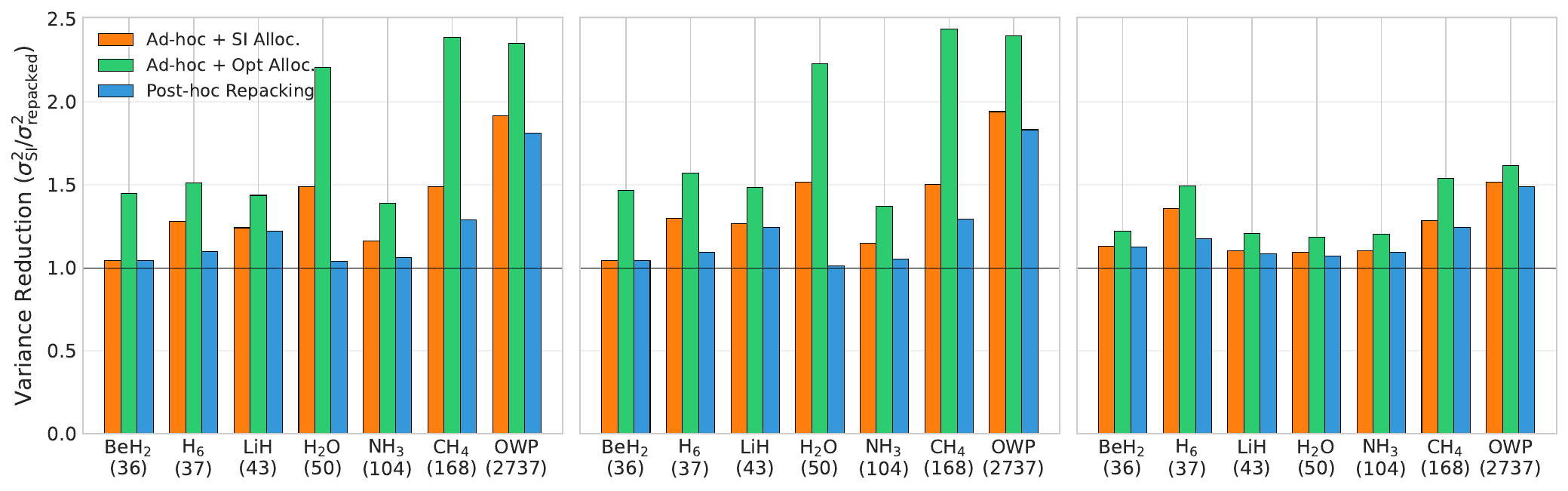}
    \caption{
        Variance reduction of overlapped groupings found through repacking relative to sorted insertion for molecular Hamiltonians. From left to right, variance is computed with respect to an approximate ground state found by DMRG with a fixed bond dimension $\chi = 64$, the Hartree-Fock state, and a random product state. For each state, we compute the variance reduction uniform allocation strategies for both sorted insertion and ad-hoc and post-hoc repacking schemes. The post-hoc strategy is described in~\eqref{eq: allocation} while the ad-hoc allocation is determined by solving the optimization problem described in~\eqref{eq: optallocation}. The horizontal axis shows molecule names (OWP = one water phosphate/metaphosphate) and, in parentheses, the number of groups found by sorted insertion. As can be seen, overlapped groupings always reduce variance. The largest variance reduction is for the largest metaphosphate molecular system, and generally the variance reduction increases for larger problems.
    }
    \label{fig:chem_variance}
\end{figure*}

In the previous section, we have analytically demonstrated the maximal variance reduction possible by overlapped grouping, and have demonstrated that overlapped groupings found by repacking never increase variance, and strictly decrease variance under mild assumptions. In this section, we are interested in evaluating the performance of overlapped groupings, in particular those found by repacking, on typical problems. We take these problems to be molecular Hamiltonians, which are commonly used to benchmark both algorithms and measurement strategies~\cite{crawfordEfficientQuantumMeasurement2021,sawayaHamLibLibraryHamiltonians2024,Sawaya_Camps_DalFavero_Tubman_Rotskoff_LaRose_2025,dalfaveroMeasurementReductionExpectation2025,choiGhost2022,choiFluidFermionicFragments2023,hadfieldMeasurementsQuantumHamiltonians2022}. Additionally, we use a model of random Hamiltonians, to be defined, to further study the practical scaling of overlapped grouping methods.

First, we consider molecular Hamiltonians to identify variance reductions in practically relevant systems. We select five common benchmark Hamiltonians and a metaphosphate model recently explored as a surrogate model for ATP hydrolysis in biochemistry~\cite{laroseCostQuantumAlgorithms2026}. This metaphosphate model acts on $44$ qubits and has $575 \cdot 10^3$ Pauli terms. For each Hamiltonian, we compute the variance without grouping, the variance with grouping found by sorted insertion, and the variance with overlapped grouping found by repacking (starting from the sorted insertion grouping). We compute each variance with respect to three states: an approximate ground state found by DRMG with a small bond dimension, the Hartree-Fock state, and a random product state. The results are shown in Fig. \ref{fig:chem_variance}. As can be seen, variance is always  reduced relative to sorted insertion, with the largest reduction coming from the largest metaphosphate molecule. As expected, the ad-hoc repacking algorithm almost always produces a larger variance reduction. This is expected because there is more flexibility when repacking, due to allowing a change of basis, than with post-hoc repacking where the basis is fixed. Interestingly, however, we see that post-hoc repacking still strictly reduces variance. We emphasize that this post-hoc overlapped grouping strategy can be performed \textit{after} samples are experimentally collected for (disjoint) grouping with no additional quantum experimental overhead, and in this sense provide somewhat of a free lunch. Additionally, we see the general trend that the reduction tends to grow with the problem size. Note that in Fig.~\ref{fig:chem_variance} the post-hoc strategy is described in~\eqref{eq: allocation} while the ad-hoc allocation is determined by solving the optimization problem described in~\eqref{eq: optallocation}.  While the post-hoc strategy can also be optimized similarly to the ad-hoc strategy, its intended use is when the shots are already taken under the sorted insertion procedure, so we follow the original allocation strategy. 

To further explore the observed behavior that the variance reduction from overlapped grouping generally increases with problem size, we consider a family of random Hamiltonians which we can easily scale. These Hamiltonians are generated by selecting $10\%$ of all Pauli strings on $n$ qubits at random and assigning them a coefficient uniformly at random from $[-1,1]$.  We compute the variance reduction with respect to a random product state. The variance reduction of (repacked) overlapped grouping relative to sorted insertion is shown in Fig.~\ref{fig:random_variance}.  Here, we again see a variance reduction which scales with problem size, and the empirical scaling for the range of problem sizes we explored is super-linear in the number of qubits.  Note that since the number of groups also grow super-linearly, this does not contradict Theorem~\ref{thm:maximal-variance-reduction}. Beyond the two families of Hamiltonians we explored here, we expect that many other Hamiltonians will exhibit this behavior. We thus expect that overlapped grouping will become more impactful for large problem sizes, in particular for Hamiltonians expected to be able to be explored on  Megaquop computers~\cite{Preskill_2025}.

\begin{figure}
    \centering
    \includegraphics[width=\columnwidth]{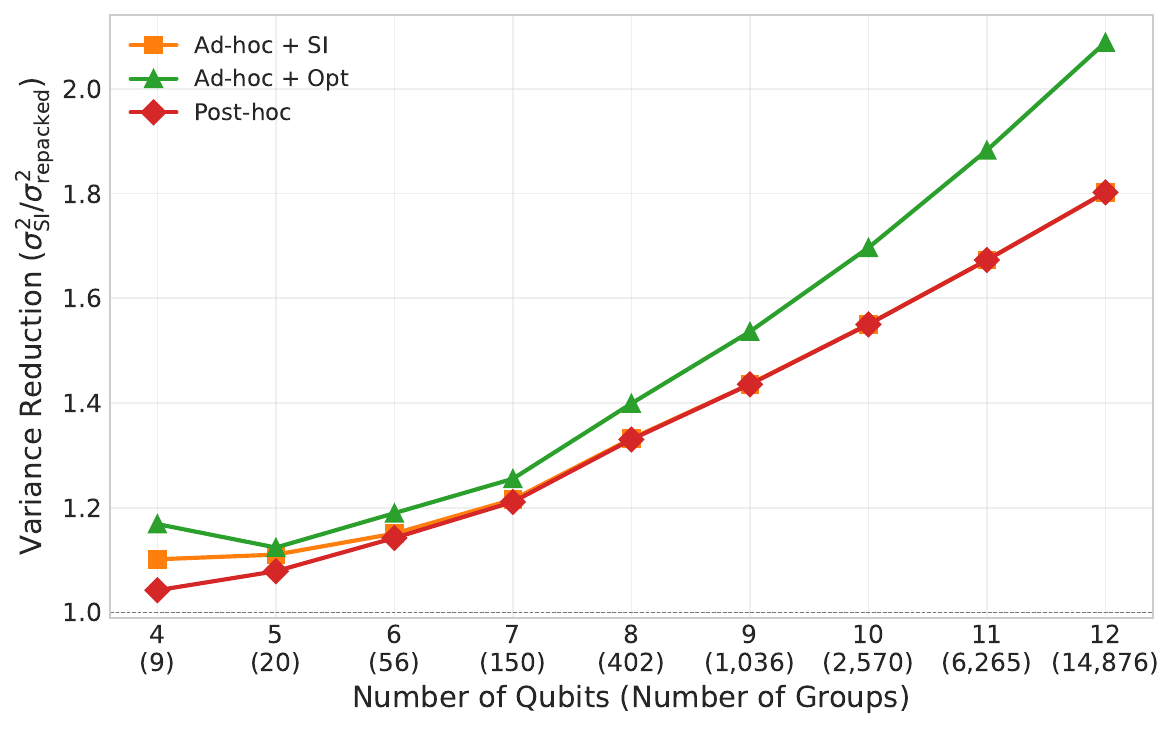}
    \caption{
        Variance reduction relative to sorted insertion with state-independent optimal allocation for random Hamiltonians as a function of qubit number.  The post-hoc method utilizes the same allocation as sorted insertion, as does the line Ad-hoc + sorted insertion.  We also consider ad-hoc repacking with convex optimized state-independent allocation which shows the best variance reduction across qubit number. Random Hamiltonians are generated by selecting $10\%$ of all Pauli strings at random and assigning them a coefficient selected uniformly at random from $[-1,1]$. The horizontal axis shows the number of qubits and the number of groups in parentheses. The largest Hamiltonian on $12$ qubits has 1.68M terms.
    }
    \label{fig:random_variance}
\end{figure}

\section{Discussion}

In Sec.~\ref{sec:results}, we showed the maximal variance reduction possible from overlapped grouping is $\Theta(m)$ relative to sorted insertion with $m$ disjoint groups, and presented an explicit repacking algorithm to achieve this. We additionally showed that repacking never increases variance, and decreases variance under mild assumptions on the Hamiltonian. These findings were corroborated by numerical experiments with electronic structure and random Hamiltonians.

In these results, we have assumed zero covariance. While this assumption has been used in prior work when constructing energy estimators and computing variances, its validity depends on both the problem Hamiltonian and the state $|\psi\rangle$ when evaluating the variance. To illustrate a system which exhibits variance-dominated behavior, we examine a two-dimensional periodic spinless Fermi-Hubbard model mapped to qubits via the Jordan-Wigner transformation.  We consider an ensemble of random product states, and exactly evaluate the variance covariance contributions to the energy~\eqref{eqn:variance-exact}. The results, shown in Fig.~\ref{fig:hubbard_varcov}, show that the total variance contribution grows linearly in the problem size, whereas the total covariance is zero on average. In this setting, the variance contribution dominates the covariance contribution, and so the zero covariance assumption is justified.

\begin{figure}[t!]
    \centering
    \includegraphics[width=\columnwidth]{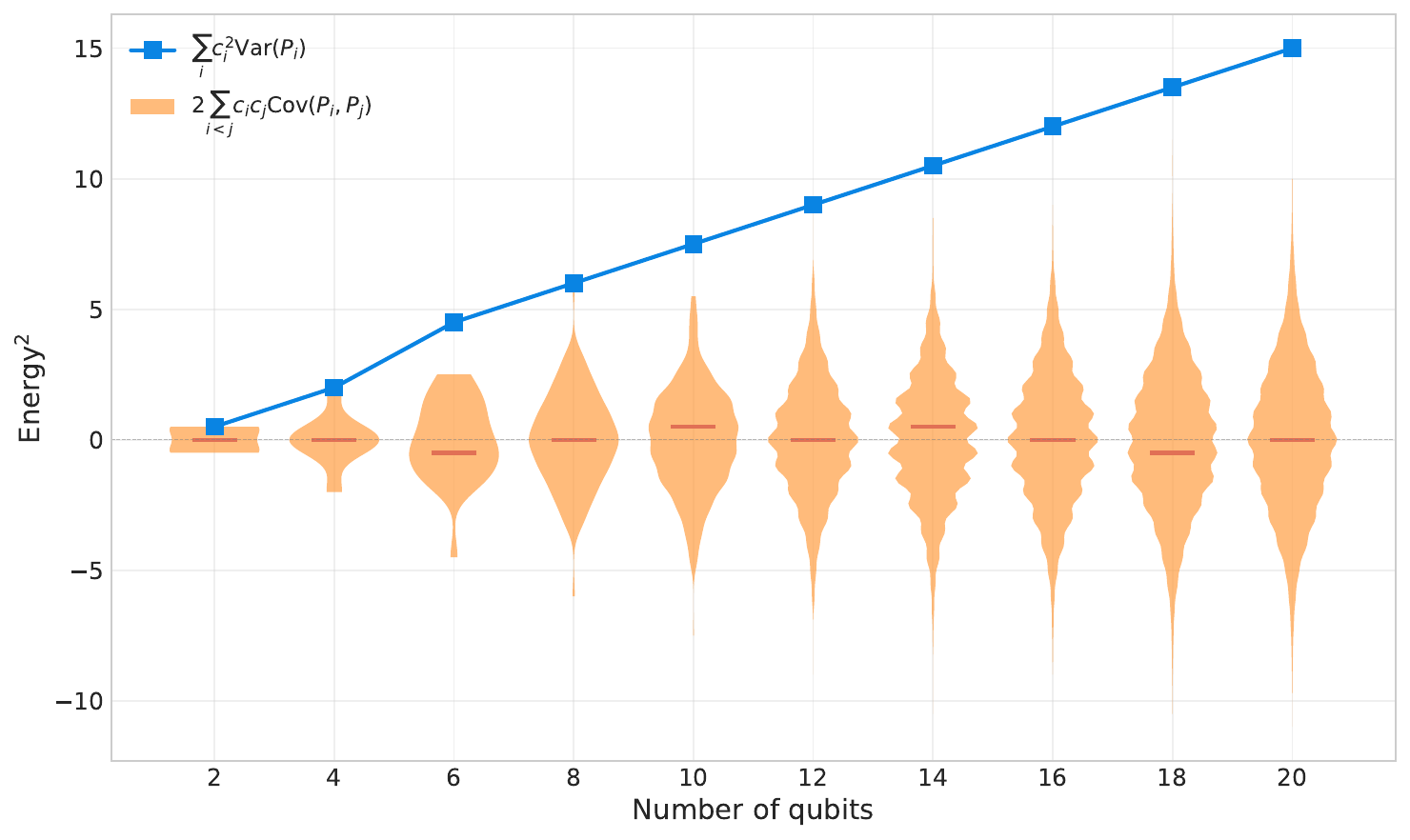}
    \caption{
        Variance (blue) vs. covariance (orange) contributions to the energy estimator for a $2 \times n$ periodic spinless Fermi-Hubbard model over an ensemble of product states.  The diagonal contribution $\sum_i c_i^2 \mathrm{Var}(P_i)$ (blue) grows linearly in the system size, while the covariance contribution $2\sum_{i<j} c_i c_j \mathrm{Cov}(P_i,P_j)$ (orange) is zero on average.  Across all system sizes, the covariance term remains smaller than the diagonal contribution, so the total variance is dominated by the diagonal component.
    }
    \label{fig:hubbard_varcov}
\end{figure}

However, it is easy to construct examples in which the opposite is true. For example, in Appendix~\ref{appendix:further-discussion-of-zero-covariance-assumption} we show that the Ising model $H_n \coloneq -\sum_{1\le i<j\le n} Z_i Z_j$ with state $|\psi\rangle = (|E_\text{min}\rangle + |E_\text{max}\rangle) / \sqrt{2}$ (where $|E_\text{min/max}\rangle$ are the min/max energy eigenstates) has a total covariance proportional to $n^3$ and a total variance proportional to $n$. In this setting, our results would not hold.

In practice, we have seen numerically (Fig.~\ref{fig:chem_variance}) that for typical electronic structure problems repacking does in fact reduce the variance of the energy estimator. This finding is also true for a model of random Hamiltonians that we studied in order to further investigate scaling (Fig.~\ref{fig:random_variance}). Finally, we remark that although we assumed zero covariance in our main results (Theorem~\ref{thm:maximal-variance-reduction-restatement} and Theorem~\ref{thm:monotonicity-zero-cov}), we remark that Theorem~\ref{thm:repacking-never-increases-variance-with-covariance} demonstrates that repacking never increases variance, without any assumptions on covariance. Again, while this is a weaker result than strictly decreasing, we have observed a strict decrease in variance in all our numerical experiments. As mentioned, the zero covariance assumption is one commonly explored in prior work. Without this assumption, even disjoint grouping methods like sorted insertion can increase the variance of the energy estimator, and this naturally extends into overlapped grouping methods as well. Two examples of this are discussed in Appendix~\ref{appendix:failure-modes-of-grouping}.

We therefore believe that overlapped groupings, in particular those found through the repacking algorithms we have introduced, will prove to be valuable techniques for energy estimation for problem sizes at the scale of Megaquop computers.

\section{Conclusion}

In this work, we proved that the maximal variance reduction possible for overlapped groupings, relative to disjoint groupings, is linear in the number of groups. Our proof introduced new algorithms, which we call repacking, for forming overlapped groupings from disjoint groupings. We have shown that these repacking algorithms never increase variance, and strictly decrease variance under mild assumptions. These findings were shown to be true numerically for electronic structure and random Hamiltonian benchmarks, in which we observed a variance reduction that is linear in the problem size.

Thus, our work characterizes the theoretical best and worst case performance of overlapped grouping methods --- an open question since the method has been suggested and numerically explored in multiple prior publications. Beyond this, we have additionally characterized the average-case performance by numerical studies on typical problems, notably for problem sizes more than twice as large as what has been studied in previous work. While many methods have been suggested for the energy estimation problem --- reflecting its timeliness and importance --- our work shows that overlapped groupings can provide a significant reduction in total measurement complexity relative to state-of-the-art methods for large-scale molecular Hamiltonians.

Several natural extensions are possible for future work. While we have started from an initial disjoint grouping to form overlapped groupings, many other methods can be explored, including those which may result in a different number of groups. Additionally, overlapped groupings provide a redundancy that may be able to be exploited for further practical benefits. Specifically, terms can be measured in multiple groups in overlapped groupings. These different groups may have different diagonalizing circuit depths, which leads to different noise strengths on quantum computers. By exploiting this redundancy, noise may be able to be mitigated through zero-noise extrapolation techniques~\cite{temmeErrorMitigationShortDepth2017,endoPracticalQuantumError2018,GiurgicaTiron_Hindy_LaRose_Mari_Zeng_2020}, leading to an ``embedded'' error mitigation property of overlapped grouping methods that is not present in disjoint groupings.

\section{Acknowledgments}

This work is supported by Wellcome Leap as part of
the Q4Bio Program. R.S was supported by the U.S. Department of Energy, Office of Science, Accelerated Research in Quantum Computing Centers, Quantum Utility through Advanced Computational Quantum Algorithms, grant no. DE-SC0025572. This work was supported in part through computational resources and services provided by the Institute for Cyber-Enabled Research (ICER) at Michigan State University.


\section*{Code and Data Availability}

Code to reproduce numerical experiments is available at~\cite{git}.





\bibliographystyle{unsrt}
\bibliography{refs}

\appendix

\section{Further discussion of zero covariance assumption} \label{appendix:further-discussion-of-zero-covariance-assumption}

Next we show that for Hamiltonians with a large spectral width, there exist states for which covariance dominates. Large spectral width is often associated with long-range order or strong global correlations, where many observables are aligned across the system.  In these regimes, a large number of Hamiltonian terms can be simultaneously extremized, leading to coherent accumulation of energy contributions and, consequently, covariance-dominated variance.

\begin{lemma}
\label{lem:LargeSpectralWidth}
Let $H=\sum_i c_i P_i$ be a Pauli Hamiltonian. For any state $\ket{\psi}$ define
\[
D(\psi) \coloneq \sum_i c_i^2\,\mathrm{Var}(P_i).
\]
Then $D(\psi)\le \|c\|_2^2$. Moreover, letting $E_{\min}$ and $E_{\max}$ denote the minimum and maximum eigenvalues of $H$ and $W\coloneq E_{\max}-E_{\min}$, there exists a state $\ket{\psi_\star}$ such that
\[
\text{Var}(H)\Big|_{\psi_\star}=\frac{W^2}{4}
\qquad\text{and}\qquad
\frac{D(\psi_\star)}{\text{Var}(H)|_{\psi_\star}} \le \frac{4\|c\|_2^2}{W^2}.
\]
\end{lemma}

\begin{proof}
Since each $P_i$ has $\pm 1$ outcomes, $\mathrm{Var}(P_i)=1-\langle P_i\rangle^2\le 1$, hence
\[
D(\psi)=\sum_i c_i^2\,\mathrm{Var}(P_i)\le \sum_i c_i^2=\|c\|_2^2.
\]
Let $\ket{E_{\min}}$ and $\ket{E_{\max}}$ be eigenstates achieving $E_{\min}$ and $E_{\max}$, and set
\[
\ket{\psi_\star}\coloneq \frac{\ket{E_{\min}}+\ket{E_{\max}}}{\sqrt{2}}.
\]
Then $\langle H\rangle_{\psi_\star}=(E_{\min}+E_{\max})/2$ and
$\langle H^2\rangle_{\psi_\star}=(E_{\min}^2+E_{\max}^2)/2$, so
\[
\text{Var}(H)\Big|_{\psi_\star}
=
\langle H^2\rangle_{\psi_\star}-\langle H\rangle_{\psi_\star}^2
=
\frac{(E_{\max}-E_{\min})^2}{4}
=
\frac{W^2}{4}.
\]
Combining with $D(\psi_\star)\le \|c\|_2^2$ gives the ratio bound.
\end{proof}

\begin{proposition}[A covariance-dominated sequence exists]
\label{prop:ising-diagonal-fraction-vanishes}
There exists a sequence of Pauli Hamiltonians $\{H_n\}$ and states $\{\ket{\psi_n}\}$ such that
\[
\frac{D_n(\psi_n)}{\mathrm{Var}(E_n)\big|_{\psi_n}}
\longrightarrow 0
\qquad\text{as } n\to\infty,
\]
where $H_n=\sum_i c_i^{(n)} P_i^{(n)}$, $D_n(\psi)\coloneq \sum_i (c_i^{(n)})^2\,\mathrm{Var}(P_i^{(n)})$, and
$E_n=\langle H_n\rangle$.
\end{proposition}

\begin{proof}
For even $n$, take
\[
H_n \coloneq -\sum_{1\le i<j\le n} Z_i Z_j.
\]
There are $M_n=\binom{n}{2}$ terms and $c_i^{(n)}=-1$, so $\|c^{(n)}\|_2^2=\sum_i (c_i^{(n)})^2=\binom{n}{2}$.
The all-aligned basis states satisfy $Z_iZ_j=+1$ for all $i<j$, hence $E_{\min}(n)=-\binom{n}{2}$.
If $n/2$ spins are $+1$ and $n/2$ are $-1$, then $\sum_{i<j} z_i z_j=-n/2$, so $E_{\max}(n)=n/2$.
Thus
\[
W_n = E_{\max}(n)-E_{\min}(n)=\frac{n^2}{2}.
\]
Let $\ket{\psi_n}$ be the witness state $\ket{\psi_\star}$ from Lemma~\ref{lem:LargeSpectralWidth} for $H_n$. Then
\[
\frac{D_n(\psi_n)}{\mathrm{Var}(E_n)\big|_{\psi_n}}
\le
\frac{4\|c^{(n)}\|_2^2}{W_n^2}
=
\frac{4\binom{n}{2}}{(n^2/2)^2}
=
\frac{8(n-1)}{n^3}
\longrightarrow 0.
\]
\end{proof}

The preceding results are stated at the level of operator variance, but they directly inform the behavior of the heuristic estimator introduced earlier. Recall that the heuristic estimator is derived under a model in which covariance terms are neglected, and its variance is governed entirely by the diagonal contribution
\[
D(\psi) = \sum_i c_i^2 \,\mathrm{Var}(P_i).
\]

The proposition shows that, for Hamiltonians with sufficiently large spectral width, there exist states for which this diagonal contribution becomes negligible compared to the total variance. In such regimes, the heuristic estimator is no longer aligned with the true variance structure of the Hamiltonian, as it ignores the dominant covariance contribution. Consequently, even if grouping and allocation are otherwise well-chosen, the estimator may perform poorly relative to covariance-aware methods in this regime.  

The bound in Lemma~\ref{lem:LargeSpectralWidth} suggests that the ratio ${\|c\|_2^2} / {W^2}$
serves as a practical diagnostic for identifying such regimes. When this ratio is small, covariance effects may dominate, and covariance-agnostic methods may perform poorly. 

\section{Failure modes of (overlapped) grouping} \label{appendix:failure-modes-of-grouping}

Here we show minimal examples where variance increases under (overlapped) grouping with the heuristic estimator under fixed shot counts.

Consider two commuting observables $A$ and $B$ measured either separately as
$\{\{A\},\{B\}\}$ or jointly as $\{\{A,B\}\}$ under a fixed total measurement budget $M$.
While joint measurement increases data reuse, it also introduces a covariance contribution
to the energy estimator.  Under equal allocation $M_A=M_B=M/2$, the joint measurement
yields a larger estimator variance whenever
\begin{equation}
\label{ineq:grouping-symmetric}
2\,\sigma_{AB} \;>\; \frac{c_A}{c_B} + \frac{c_B}{c_A},
\end{equation}
where $\sigma_{AB}=\langle AB\rangle-\langle A\rangle\langle B\rangle$.
This effect arises purely from the covariance structure induced by grouping and does not rely
on repacking.  Repacking inherits this estimator-level limitation, but does not introduce it.

Indeed, consider 
Let $H=c_A A+c_B B+c_C C$ with $[A,B]=[A,C]=0$.  Consider the grouping
\[
\mathcal G=\bigl\{\{A,C\},\{B\}\bigr\},
\]
measured with $M_1$ shots on $\{A,C\}$ and $M_2$ shots on $\{B\}$.  Repacking $A$ into the second group gives
\[
\mathcal R=\bigl\{\{A,C\},\{A,B\}\bigr\},
\]
measured with the same shot counts $(M_1,M_2)$.

Recall from Sec.~\ref{sec:definitions} that if an observable $P$ is measured $M$ times, its shot-noise variance is
$\sigma^2_P/M$ with $\sigma_P:=1-\langle P\rangle^2$, and if $P$ and $Q$ are measured jointly $M$ times, their shot-noise
covariance is $\sigma_{PQ}/M$ with $\sigma_{PQ}:=\langle PQ\rangle-\langle P\rangle\langle Q\rangle$. Denote the estimator for the energy before and after repacking as $\overline{E}_{\mathcal{G}}$ and $\ovl{E}_{\Rcal}$ respectively. Using the shot-weighted averaging scheme for the weights, a simple calculation then yields the following expressions:
\begin{equation*}
\Vr{\overline{E}_{\mathcal{G}}} = \frac{1}{M_1} \left(c_A^2 \sigma^2_A + c_C^2 \sigma^2_C + 2 c_A c_C \sigma_{A C} \right) + \frac{c_B^2 \sigma^2_B}{M_2},
\end{equation*}
and
\begin{equation*}
\begin{split}
\Vr{\overline{E}_{\mathcal{R}}} &= \frac{c_A^2 \sigma^2_A}{M_1 + M_2} + \frac{c_C^2 \sigma^2_C }{M_1} + \frac{2 c_A c_C \sigma_{A C}}{M_1 + M_2} \\
&+ \frac{c_B^2 \sigma^2_B}{M_2} + \frac{2 c_A c_B \sigma_{A B}}{M_1 + M_2}.
\end{split}
\end{equation*}
Subtracting yields the following for $\text{Var}(\overline E_{\mathcal R})-\text{Var}(\overline E_{\mathcal G})$:
\begin{equation*}
\begin{split}
\frac{2 c_A c_B \sigma_{AB}}{M_1 + M_2} - \frac{M_2 \left( c_A^2 \sigma^2_A + 2 c_A c_C \sigma_{AC} \right)}{M_1 (M_1 + M_2)}.
\end{split}
\end{equation*}
Therefore $\text{Var}(\overline E_{\mathcal R})>\text{Var}(\overline E_{\mathcal G})$ if and only if
\begin{equation*}
\label{eq:increase-variance-example}
2 c_A c_B \sigma_{AB} - \frac{M_2}{M_1} \left( c_A^2 \sigma^2_A + 2 c_A c_C \sigma_{AC} \right) > 0,
\end{equation*}
or equivalently,
\begin{equation}
\label{eq:increase-variance-example-1}
2 c_A c_B \sigma_{AB} > \frac{M_2}{M_1} \left( c_A^2 \sigma^2_A + 2 c_A c_C \sigma_{AC} \right).
\end{equation}
For fixed $A, B, C$ and initial state $\ket{\psi}$, the quantities $\sigma^2_A, \sigma_{AB}, \sigma_{AC}$ are fixed. Thus it is easy to find situations when Eq.~\eqref{eq:increase-variance-example-1} is satisfied. In fact, one can make the left hand side arbitrarily large by letting $c_B \rightarrow \infty$, fixing everything else, and assuming $c_A, \sigma_{AB} > 0$.

To get a more easily understandable implication of the above proposition, we can simplify Eq.~\eqref{eq:increase-variance-example-1} under the assumptions $\sigma^2_A = \sigma_{AB} = 1$, and $\sigma_{AC}=0$, which yields
\begin{equation}
\label{eq:covinc}
    \frac{2 M_1}{M_2} > \frac{c_A}{c_B}.
\end{equation}
From Eq.~\eqref{eq:covinc}, a variance increase under the heuristic estimator under repacking is inherently asymmetric in the two operators involved.  The asymmetry reflects the directional nature of repacking: one operator is inserted into an existing measurement context associated with another, and the resulting estimator normalization and shot allocation treat the two terms differently.  As a result, unfavorable covariance effects arise only when covariance contribution, amplified by relative undersampling, overwhelm the coefficient imbalance between the operators.  Several features of the sorted insertion and ad-hoc repacking strategies mitigate these scenarios in practice.  First, both heuristics prioritize inserting operators with larger coefficients, so typically $c_A > c_B$.  Second, optimal allocation strategies based on worst-case bounds apply more shots to groups with large coefficient mass, so typically $M_1>M_2$.  Finally, across Hamiltonians we have studied, covariance contributions arising from repacking involve a variety of operators pairs which induce covariances of different sign and magnitude, so individual covariance terms may increase or decrease the overall estimator variance, but their combined contribution does not typically overwhelm the leading order variance reduction obtained from increased effective sampling.

Note that the previous example shows an interesting feature.  It is possible to find situations where Eq.~\eqref{eq:increase-variance-example-1} is never satisfied: for example, suppose that $\sigma_{AB}, \sigma_{AC} > 0$, $c_A, c_C >0$, and $c_B < 0$, in which case the right hand side is positive and the left hand side is negative. In this case, the shot-weighted averaging estimator will always reduce the variance after repacking in the example discussed there.

\end{document}